\documentclass[11pt]{article}
\usepackage[margin=1in]{geometry}

\usepackage[T1]{fontenc}
\usepackage[utf8]{inputenc}
\usepackage{lmodern}
\usepackage{microtype}
\usepackage{amsmath,amssymb,amsthm,mathtools,bbm}
\usepackage{enumitem}
\usepackage{multirow}
\usepackage{algorithm}
\usepackage{caption}

\usepackage{tikz}
\usepackage{graphicx}
\usepackage{pgfplots}
\pgfplotsset{compat=1.18}
\usetikzlibrary{patterns}
\usepackage{algpseudocode}
\algtext*{EndProcedure}
\usetikzlibrary{calc,arrows.meta,positioning}

\usepackage{booktabs}
\usepackage{makecell}
\usepackage{array}

\usepackage{xcolor}
\definecolor{utburnt}{HTML}{BF5700}
\definecolor{tealblue}{HTML}{007C91}

\usepackage{natbib}
\usepackage[
  colorlinks=true,
  linkcolor=utburnt,
  citecolor=tealblue,
  urlcolor=tealblue
]{hyperref}
\usepackage[nameinlink,capitalize]{cleveref}
\usepackage{aliascnt}

\theoremstyle{plain}
\newtheorem{theorem}{Theorem}[section]

\newaliascnt{proposition}{theorem}
\newtheorem{proposition}[proposition]{Proposition}
\aliascntresetthe{proposition}

\newaliascnt{lemma}{theorem}
\newtheorem{lemma}[lemma]{Lemma}
\aliascntresetthe{lemma}

\newaliascnt{corollary}{theorem}
\newtheorem{corollary}[corollary]{Corollary}
\aliascntresetthe{corollary}

\newaliascnt{fact}{theorem}
\newtheorem{fact}[fact]{Fact}
\aliascntresetthe{fact}

\theoremstyle{remark}
\newaliascnt{remark}{theorem}
\newtheorem{remark}[remark]{Remark}
\aliascntresetthe{remark}

\theoremstyle{definition}
\newaliascnt{definition}{theorem}
\newtheorem{definition}[definition]{Definition}
\aliascntresetthe{definition}

\Crefname{theorem}{Theorem}{Theorems}
\Crefname{proposition}{Proposition}{Propositions}
\Crefname{lemma}{Lemma}{Lemmas}
\Crefname{corollary}{Corollary}{Corollaries}
\Crefname{fact}{Fact}{Facts}
\Crefname{remark}{Remark}{Remarks}
\Crefname{definition}{Definition}{Definitions}

\setlist[itemize]{topsep=0.25em,itemsep=0.15em,leftmargin=1.5em}
\setlist[enumerate]{topsep=0.25em,itemsep=0.15em,leftmargin=1.8em}

\DeclareRobustCommand{\NazBound}{\cite{Naz03}}

\numberwithin{equation}{section}

\newcommand{\Ind}{\mathbbm{1}}

\newcommand{\R}{\mathbb{R}}
\DeclareMathOperator*{\E}{\mathbb{E}}
\newcommand{\eps}{\varepsilon}
\newcommand{\conv}{\mathrm{conv}}
\newcommand{\OPT}{\mathrm{OPT}}
\newcommand{\err}{\mathrm{err}}
\newcommand{\sign}{\mathrm{sign}}
\newcommand{\rank}{\mathrm{rank}}
\newcommand{\GSA}{\mathrm{GSA}}
\newcommand{\poly}{\mathrm{poly}}

\newcommand{\Var}{\mathrm{Var}}
\newcommand{\tr}{\mathrm{tr}}
\newcommand{\ip}[2]{\langle #1,#2\rangle}

\newcommand{\cG}{\mathcal{G}}
\newcommand{\cP}{\mathcal{P}}

\newcommand{\cL}{\mathcal{L}}
\newcommand{\cH}{\mathcal{H}}
\newcommand{\cN}{\mathcal{N}}

\usepackage{xstring}

\newcommand{\algoname}[1]{%
  \IfEqCase{#1}{%
    {1}{Agnostic Proper Learner for Halfspaces}%
  }[Unknown algorithm]%
}

\title{Proper Agnostic Learning of Functions of Halfspaces\\ under Gaussian Marginals}

\author{
 \begin{tabular}{cc}
        \begin{tabular}{c}
            Sergei Tikhonov\\ \texttt{tikhonov@utexas.edu} \\ UT Austin 
        \end{tabular} & 
         \begin{tabular}{c}
              Arsen Vasilyan\thanks{Supported by the NSF AI Institute for Foundations of Machine Learning (IFML).} \\ \texttt{arsenvasilyan@gmail.com} \\ UT Austin
         \end{tabular}
    \end{tabular}
}

\begin{document}

\maketitle

\begin{abstract}
  We study the problem of computationally efficient \emph{proper agnostic learning} of multidimensional concept classes under the Gaussian distribution. In this setting, given i.i.d.\ labeled samples from an unknown distribution over $\mathbb{R}^d \times \{\pm 1\}$ whose marginal on $\mathbb{R}^d$ is Gaussian, the goal is to output a hypothesis from a target class $\mathcal{F}$ whose 0-1 loss is within $\eps$ of that of the best classifier in $\mathcal{F}$.

We give the first efficient proper agnostic learning algorithm for arbitrary Boolean functions of $K$ halfspaces under Gaussian marginals. Our algorithm runs in time $d^{O(K^2 \log(1/\eps)/\eps^2)} + (K/\eps)^{O(K^3/\eps^{2.5})}$. Prior to our work, the only known algorithm for $K \geq 2$ was brute-force search, with run-time exponential in $d$. Moreover, the dependence of our run-time on the dimension $d$ matches that of the best known \emph{improper} learning algorithm, namely $d^{\widetilde{O}(K^2/\eps^2)}$.

For the special case of a single halfspace ($K=1$), the best previous run-time was $d^{O(1/\eps^4)} + (1/\eps)^{O(1/\eps^6)}$.\ Our algorithm improves this to $d^{O(1/\eps^2)} + (1/\eps)^{O(1/\eps^{2.5})}$. Once again, the dependence on $d$ matches that of the best known improper algorithm, namely $d^{O(1/\eps^2)}$. Furthermore, the dependence of our run-time on the dimension $d$ is essentially optimal in the statistical query model.

\end{abstract}

\section{Introduction}

In this work, we focus on computationally efficient \emph{proper agnostic learning} of multidimensional concepts such as halfspaces and Boolean functions of halfspaces. Proper agnostic learning is the fundamental task of \emph{fitting an approximately optimal classifier} from a hypothesis class $\mathcal{F}$ to an i.i.d. labeled dataset $S$. More formally:
\begin{definition}[Proper Agnostic Learning]
    Let $\mathcal{F}$ be a hypothesis class. Given an i.i.d. sample set $S=\{(\mathbf{x}_i,y_i)\}$ drawn from a distribution $D$, where $\mathbf{x}_i \in \R^d$ and $y_i \in \{\pm 1\}$, a proper agnostic learning algorithm for $\mathcal{F}$ must produce a hypothesis $h \in \mathcal{F}$ satisfying
    $$
    \Pr_{(\mathbf{x},y)\sim D}[h(\mathbf{x})\neq y]
    \leq
    \OPT_{\mathcal{F}} 
    +\eps
    $$
    with high probability, where $\OPT_{\mathcal{F}}:=\inf_{h'\in \mathcal{F}}
    \Pr_{(\mathbf{x},y)\sim D}[h'(\mathbf{x})\neq y]$.
\end{definition}
Fitting a near-optimal classifier from a hypothesis class, when no classifier fits the data perfectly, is one of the most basic goals in machine learning. Yet it is well known that without any assumption on the distribution $D$, this task is NP-hard even for the elementary class of \emph{halfspaces} (also known as \emph{linear classifiers}) \cite{guruswami2006hardness, feldman2006new}, and is therefore believed to require run-time exponential\footnote{This contrasts sharply with the realizable setting, in which the data is linearly separable and efficient algorithms have been known for decades \cite{rosenblatt1958perceptron}.} in the data dimension $d$. We therefore focus on the distribution-specific setting, where an assumption is placed on the distribution of the data points $\{\mathbf{x}_i\}$. Specifically, we adopt the standard assumption that $\{\mathbf{x}_i\}$ are drawn from a Gaussian distribution (see e.g.\ \cite{KKMS08,vempala2010learningconvex,vempala2010randomsampling,diakonikolas2018learning,awasthi2017power,diakonikolas2022learning, diakonikolas2020non,DiakonikolasKaneKontonisTzamosZarifis2021}).

Our approach applies to the concept class of \emph{Boolean functions of $K$ halfspaces}, yielding the first efficient proper agnostic learning algorithm for this hypothesis class under the Gaussian distribution. Specifically, our algorithm runs in time $d^{O(K^2\log(1/\eps)/\eps^2)}
+
(K/\eps)^{
O\!\left(K^3/\eps^{2.5}\right)}$ and returns a Boolean function of $K$ halfspaces whose accuracy is within $\eps$ of the optimal accuracy among all Boolean functions of $K$ halfspaces, which we denote as $\OPT_K$.
Previously, for $K\geq 2$, the only available algorithm for proper learning of this concept class was the exponential-time algorithm based on brute-force search.

We emphasize that even for the most basic class of multidimensional concepts --- the class of halfspaces (a.k.a. linear classifiers) --- our work yields significant improvements over what was known previously. The prior state-of-the-art algorithm of \cite{DiakonikolasKaneKontonisTzamosZarifis2021} runs in time $d^{O(1/\eps^4)}
+
(1/\eps)^{O(1/\eps^6)}$, where $d$ is the data dimension and $\eps$ is the accuracy parameter. Our algorithm for this task has the improved run-time of $d^{O(1/\eps^2)}
+
(1/\eps)^{O(1/\eps^{2.5})}$. In particular, the dependence of our run-time on the dimension $d$ is likely to be optimal as it matches the state-of-the-art improper learning run-time $d^{O(1/\eps^2)}$ \cite{KKMS08, DKN10}, as well as the run-time lower bound of $d^{\Omega(1/\eps^2)}$ for statistical query algorithms \cite{diakonikolas2020near, goel2020statistical, pmlr-v134-diakonikolas21c}.

In addition to proper agnostic learning of halfspaces and Boolean functions of halfspaces, our approach extends to \emph{intersections of halfspaces} --- also yielding the first efficient proper agnostic learning algorithm under the Gaussian distribution for this hypothesis class. Overall, we summarize our run-times in Table \ref{tab:proper-learning-comparison}.

\begin{table*}[h]
\centering
\scriptsize
\renewcommand{\arraystretch}{1.25}
\setlength{\tabcolsep}{6pt}
\begin{tabular}{@{}
    >{\centering\arraybackslash}p{2.8cm}
    >{\centering\arraybackslash}p{7.0cm}
    >{\centering\arraybackslash}p{3.2cm}
@{}}
\toprule
\textbf{Concept class} & \textbf{This Work} & \textbf{Prior Best} \\

\midrule

\makecell[c]{Boolean functions of\\ $K$ halfspaces}
&
$d^{O(K^2\log(1/\eps)/\eps^2)}
+
(K/\eps)^{
O\!\left(K^3/\eps^{2.5}\right)}$
&
\makecell[c]{$(1/\eps)^{\poly(dK)}$\\ (brute force search)}
\\
\midrule

\makecell[c]{Intersections of\\ $K$ halfspaces}
&
$d^{O(\log K\,\log(1/\eps)/\eps^2)}
+
\left(\sqrt{\log K}/\eps\right)^{
O\!\left(K^2\sqrt{\log K}/\eps^{2.5}\right)}$
&
\makecell[c]{$(1/\eps)^{\poly(dK)}$\\ (brute force search)}

\\
\midrule

\makecell[c]{Halfspaces}
&
$d^{O(1/\eps^2)}
+
(1/\eps)^{O(1/\eps^{2.5})}$
&
\makecell[c]{$d^{O(1/\eps^4)}
+
(1/\eps)^{O(1/\eps^6)}$ \\ \cite{DKTKZ21}}
\\

\bottomrule
\end{tabular}
\caption{Run-time comparison for proper agnostic learning under Gaussian marginals.}
\label{tab:proper-learning-comparison}
\end{table*}

\paragraph{Comparison with improper learning algorithms.}
Our run-time for proper agnostic learning of functions of $K$ halfspaces is close to matching the $d^{\tilde{O}(K^2/\eps^2)}$ run-time of the state-of-the-art \emph{improper} learning algorithm of \cite{KKMS08, KOS08, PSW26}. That algorithm achieves error $\OPT_{K}+\eps$, but the classifier it produces is a Polynomial Threshold Function (PTF), i.e.\ of the form $\sign(P(\mathbf{x}))$ where $P$ is a low-degree multivariate polynomial\footnote{A learning algorithm for a hypothesis class $\mathcal{F}$ is called \emph{improper} if the hypothesis it produces need not itself lie in $\mathcal{F}$.}.

Such improper learning algorithms have the obvious drawback of outputting a highly complex function---a PTF---whose accuracy is only guaranteed to match that of the best classifier in a much simpler hypothesis class (here, functions of $K$ halfspaces). Moreover, the use of PTFs incurs overhead in storage, evaluation time, and adversarial robustness of the decision boundary \cite{lange2025robust}.

Due to these limitations, matching the run-time of improper learning algorithms with proper ones has been a major research direction in computational learning theory \cite{DiakonikolasKaneKontonisTzamosZarifis2021, lange2022properly, lange2025agnostic, pinto2025learning, blanc2022properly}. We come close to achieving this for functions of halfspaces: the dependence of our run-time $d^{O(K^2\log(1/\eps)/\eps^2)} + (K/\eps)^{O\!\left(K^3/\eps^{2.5}\right)}$ on the dimension $d$ matches that of the best known \emph{improper} learning algorithm, $d^{O(K^2\log(1/\eps)/\eps^2)}$ \cite{KKMS08, KOS08, PSW26}.

\begin{figure}[h!]
\centering
\fcolorbox{black}{gray!6}{\resizebox{\dimexpr\textwidth-2\fboxsep-2\fboxrule\relax}{!}{%
\begin{tikzpicture}

  % Red hatched hardness region
  \fill[red!15] (-3,-0.4) rectangle (2,0.4);
  \fill[pattern=north east lines, pattern color=red!55]
      (-3,-0.4) rectangle (2,0.4);
  \node[red, scale=1.15, align=center] at (-0.8,1.0)
      {Computational \\ hardness};

  % Run-time axis
  \draw[->, >=stealth, line width=3pt] (-3,0) -- (13.5,0)
      node[right] {\large Run-time};

  % Tick marks (improper and this work coincide at x=2)
  \draw[line width=2pt] (2, 0.3)  -- (2, -0.3);
  \draw[blue, line width=2pt] (2, 0.3) -- (2, -0.3);
  \draw[line width=2pt] (8, 0.3)  -- (8, -0.3);

  % Tick labels
  \node[anchor=north west, xshift=-0.8cm, yshift=-14pt,
        scale=1.15, align=left] at (2,0)
      {$d^{{O}(1/\eps^{2})}$ \\
       Improper learning \\
       {\cite{KKMS08, DKN10}}};

  \node[blue, anchor=south west, xshift=-0.8cm, yshift=14pt,
        scale=1.15, align=left] at (2,0)
      {$d^{{O}(1/\eps^{2})}$ \\
       Proper learning \\
       \textbf{[This work]}};

  \node[anchor=north west, xshift=-0.8cm, yshift=-8pt,
        scale=1.15, align=left] at (8,0)
      {$d^{\tilde{O}(1/\eps^{4})}$ \\
       Proper learning (previous best) \\
    {\cite{DiakonikolasKaneKontonisTzamosZarifis2021}}};

\end{tikzpicture}}}%
\caption{Agnostic proper learning of a single halfspace from Gaussian data in the \emph{moderate-accuracy regime} %$\log d \geq 1/\eps^{C}$..
($\eps \geq  \frac{1}{\log(d)}$). Our algorithm improves over the previous best proper agnostic learning run-time and matches the best known improper learning run-time. The highlighted region indicates run-time lower bounds for statistical query algorithms \cite{diakonikolas2020near, goel2020statistical, pmlr-v134-diakonikolas21c}.}
  \label{fig:extreme-high-dim-regime}
\end{figure}
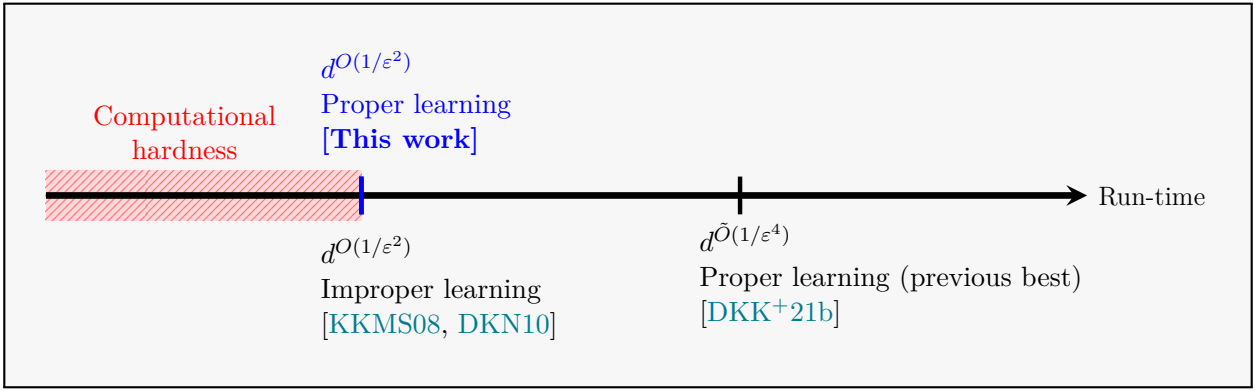

\begin{figure}[h!]
\centering
\fcolorbox{black}{gray!6}{\resizebox{\dimexpr\textwidth-2\fboxsep-2\fboxrule\relax}{!}{%
\begin{tikzpicture}

  \fill[red!15] (-3,-0.4) rectangle (2,0.4);
  \fill[pattern=north east lines, pattern color=red!55]
      (-3,-0.4) rectangle (2,0.4);
  \node[red, scale=1.15, align=center] at (-0.8,1.0)
      {Computational \\ hardness};

  \draw[->, >=stealth, line width=3pt] (-3,0) -- (13.5,0)
      node[right] {\large Run-time};

  \draw[line width=2pt] (2, 0.3)  -- (2, -0.3);   % improper
  \draw[blue, line width=2pt] (3, 0.3) -- (3, -0.3); % this work
  \draw[line width=2pt] (12, 0.3) -- (12,-0.3);   % prior proper
  
  \node[anchor=north west, xshift=-0.8cm, yshift=-14pt,
        scale=1.15, align=left] at (2,0)
      {$2^{\tilde{O}(1/\eps^{2})}$ \\
       Improper learning \\
       {\cite{KKMS08, DKN10}}};

  \node[blue, anchor=south west, xshift=-0.8cm, yshift=8pt,
        scale=1.15, align=left] at (3,0)
      {$2^{\tilde{O}(1/\eps^{2.5})}$ \\
       Proper learning \\
       \textbf{[This work]}};

  \node[anchor=north west, xshift=-0.8cm, yshift=-8pt,
        scale=1.15, align=left] at (12,0)
      {$2^{\tilde{O}(1/\eps^{6})}$ \\
       Proper learning \\\cite{DiakonikolasKaneKontonisTzamosZarifis2021}};

\end{tikzpicture}}}
\caption{Agnostic proper learning of a single halfspace from Gaussian data in the \emph{high-accuracy regime} ($\eps \leq \frac{1}{\log^2(d)}$). Our algorithm improves over the previous best proper agnostic learning run-time and comes close to matching the best known improper learning run-time. The highlighted region indicates run-time lower bounds for statistical query algorithms \cite{diakonikolas2020near, goel2020statistical, pmlr-v134-diakonikolas21c}.}
  \label{fig:moderately-high-dim-regime}
\end{figure}
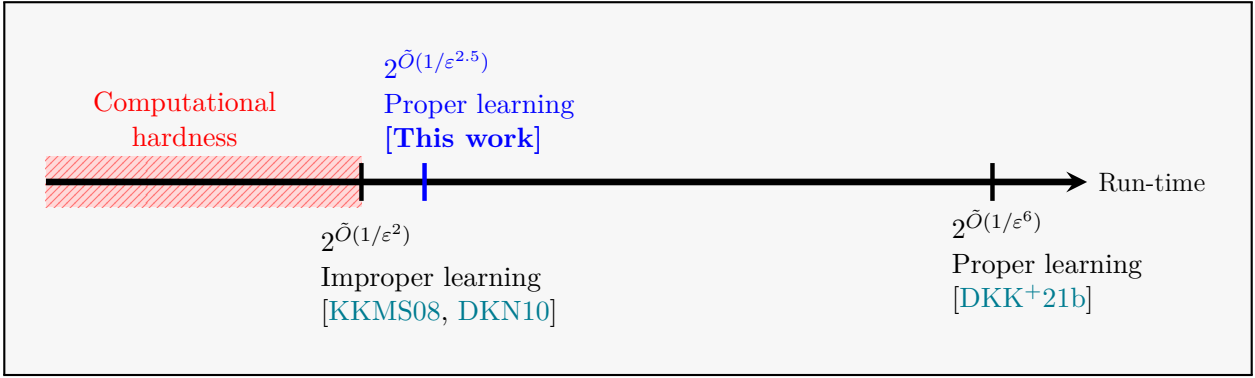
\paragraph{Comparison to computational hardness results.}
Based on statistical query lower bounds of \cite{diakonikolas2020near, goel2020statistical, pmlr-v134-diakonikolas21c}, it is believed that any agnostic learning algorithm for the class of halfspaces under the Gaussian distribution needs to run in time at least $d^{\Omega(1/\eps^2)}$. This belief is further substantiated by run-time lower bounds based on cryptographic hardness assumptions \cite{pmlr-v202-diakonikolas23b}. Therefore, the run-time of our algorithm for proper agnostic learning of a single halfspace is close to matching not only that of the state-of-the-art improper agnostic learning algorithm \cite{KKMS08,DKN10}, but also the run-time lower bounds for the agnostic learning of halfspaces under Gaussian data points. 

Overall, see Figures \ref{fig:extreme-high-dim-regime} and \ref{fig:moderately-high-dim-regime} for a summary of how our run-time compares with the previous best proper agnostic run-time of \cite{DKTKZ21}, improper learning run-time \cite{KKMS08,DKN10} and computational intractability results for statistical query algorithms \cite{diakonikolas2020near, goel2020statistical, pmlr-v134-diakonikolas21c}.
\begin{itemize}
    \item 
We see that in the moderate-accuracy regime, where $\eps \geq  \frac{1}{\log(d)}$, our run-time matches the run-time of the best improper learning algorithm \cite{KKMS08, DKN10} and statistical query lower bounds (See Figure \ref{fig:extreme-high-dim-regime}).

\item 
In the high-accuracy regime, where $\eps \leq  \frac{1}{\log^2(d)}$, our run-time is $2^{\tilde{O}(1/\eps^{2.5})}$, improving on the run-time of $2^{\tilde{O}(1/\eps^{6})}$ of \cite{DiakonikolasKaneKontonisTzamosZarifis2021} and coming close to matching the run-time $2^{\tilde{O}(1/\eps^{2})}$ of the improper learning algorithm \cite{KKMS08, DKN10}.
\end{itemize}

Likewise, for the class of intersections of $K$ halfspaces, our run-time comes close to matching the state-of-the-art run-time lower bounds of $d^{\widetilde{\Omega}(\log K)}$ for  statistical query due to \cite{DKPZ21, HsuSSV22}.

\subsection{Our Contributions}

Our main theorem is the first proper learning algorithm for Boolean functions of halfspaces under Gaussian marginals.

\begin{theorem}[Proper learning of Boolean functions of halfspaces]\label{thm:main-B}
There exists a proper agnostic learner for the class of arbitrary Boolean functions of halfspaces under Gaussian marginals with the following guarantee. For every $K\ge 1$, every $\eps,\delta\in(0,1/2)$, and every distribution $D$ over $(\mathbf{x},y)\in\R^d\times\{\pm1\}$ whose $\mathbf{x}$-marginal is $\cN(0,\mathbf{I}_d)$, the learner draws
$
N = d^{O(K^2\log(1/\eps)/\eps^2)}\,\log(1/\delta)
$
samples, runs in time
$
d^{O(K^2\log(1/\eps)/\eps^2)}\,\log(1/\delta) +
(K/\eps)^{O(K^3/\eps^{2.5})}\,\log(1/\delta),
$
and returns a hypothesis $h\in\mathcal{C}_{K}$ such that, with probability at least $1-\delta$, 
$
\err_D(h)\le \OPT_{K}+\eps.
$
\end{theorem}

For Boolean functions of $K$ halfspaces, the improper learner by Klivans et al. \cite{KOS08} runs in time $d^{O(K^2/\eps^4)}$, while the result of Pesenti et al. \cite{PSW26} improves this to $d^{\widetilde{O}(K^2/\eps^2)}$. This matches the dimension dependence in the running time of our proper learner.

Moreover, in the case of a single halfspace $(K = 1)$, we give an algorithm with a slightly faster time $d^{O(1/\eps^2)}$ (\Cref{thm:main}), matching the best known improper polynomial regression exponent under Gaussian marginals \cite{KKMS08,DKN10}. See \Cref{app:intersections} for proper agnostic learning of intersections of halfspaces.
 %Moreover, this rate is essentially optimal in the statistical query model \cite{Kea98}: Diakonikolas et al. \cite{DKPZ21} obtained optimal SQ lower bounds for agnostically learning halfspaces under Gaussian marginals.

\subsection{Algorithm overview}

For the sake of simplicity, let us first focus on the problem of proper agnostic learning of a single halfspace. We discuss the more general case of Boolean functions of halfspaces in  \Cref{sec:extension-functions-halfspaces}. Recall that the previous best run-time for this task was achieved by the algorithm of \cite{DiakonikolasKaneKontonisTzamosZarifis2021} that follows the outline below:
\begin{enumerate}
    \item Compute a degree-$1/\eps^4$ polynomial $P$ that approximately minimizes the $L_2$ error, i.e. minimizing $\sqrt{\E_{\mathbf{x},y \sim D}[(y - P(\mathbf{x}))^2]}$.
    \item Compute a matrix  $\mathbf{M}(P)=\E[\nabla P(\mathbf{x})\nabla P(\mathbf{x})^\top]$, and set $V$ to be the span of eigenvectors of $\mathbf{M}(P)$ whose eigenvalues are sufficiently large.
    \item Perform a brute-force search among all halfspaces with normals in $V$ and select the most accurate halfspace in this collection (as measured on a validation set). 
\end{enumerate}

A key obstacle for improving the run-time of \cite{DiakonikolasKaneKontonisTzamosZarifis2021} is that the first step above is based on $L_2$ regression, and one \textbf{needs} degree-$1/\eps^4$ polynomial to $\eps$-approximate a halfspace in $L_2$ error. Yet, even writing down a degree-$1/\eps^4$ polynomial requires $d^{\Omega(1/\eps^4)}$ time and space.

Instead of using the $L_2$ loss, we use the following regularized convex objective
\begin{equation}\label{logistic-loss}
\cL_{\mu,\nu}(P)
:=
\E[\ell(yP(\mathbf{x}))]
+
\mu\E[P(\mathbf{x})^2]
+
\nu\,\tr\big(\mathbf{M}(P)^{1/2}\big),
\end{equation}
where the loss term $\E[\ell(yP(\mathbf{x}))]$ is chosen so that its negative derivative\footnote{The loss $\ell(u)$ is decreasing in the margin $u$, $\ell'(u)\le 0$. We define $\psi(u):=-\ell'(u)$ so that $\psi$ is nonnegative.} $\psi(u):=-\ell'(u)$ satisfies the following properties:
\begin{enumerate}
    \item \textbf{Boundedness:} $\psi$ is uniformly bounded.
    \item \textbf{Lipschitzness:} $\psi$ is Lipschitz.
    \item \textbf{Label-flip symmetry:}
    $
    \psi(yP(\mathbf{x}))+\psi(-yP(\mathbf{x}))=1.
    $
\end{enumerate}

Boundedness of the negative derivative $\psi$ makes sure the function $\ell$ does not grow too fast. At a high level, this allows us to argue the correctness of our algorithms using the existence of a polynomial that approximates halfspaces in $L_1$ norm (and satisfies a number of other technical conditions).
This, in turn, allows us to use a lower-degree polynomial because, unlike the case of the $L_2$ norm, approximating halfspaces under the Gaussian in $L_1$ norm (\Cref{lem:q-approx}) requires only degree $O(1/\eps^2)$. 

Lipschitzness of $\psi$ provides the needed smoothness; in particular, it allows us to control the change in the objective under small perturbations of the polynomial $P$, which is an important technique in our analysis (\Cref{lem:no-large-gain}). The label-flip symmetry gives a useful decomposition of the loss derivative (\Cref{lem:label-flip-decomposition}): it separates the label-dependent term from the term depending only on the fitted polynomial $P$.

The logistic loss $\ell(u):=\log(1+e^{-u})$ \cite{Berkson1944,Cox1958}, whose negative derivative is $\psi(u):=-\ell'(u)=1/(1+e^u)$, satisfies the boundedness, Lipschitzness, and label-flip symmetry conditions. Other losses, including the probit loss and the Huberized hinge loss, also fit this framework. In this paper, we use the logistic loss, allowing us to work with a classical and well-studied convex surrogate. 

The term $\E[P(\mathbf{x})^2]$ is included mainly for making the objective strongly convex; see Appendix~\ref{app:compute_reg}. This rules out the well-known pathology of unregularized logistic regression in separable or degenerate cases, where the infimum is approached only as the coefficient norm diverges \cite{AlbertAnderson1984}.

The term $\tr(\mathbf{M}(P)^{1/2})$ controls the spectrum of $\mathbf{M}(P)=\E[\nabla P(\mathbf{x})\nabla P(\mathbf{x})^\top]$. Since the final search subspace $V$ is defined by the large-eigenvalue directions of $\mathbf{M}(P)$, this regularizer helps keep the search space low-dimensional; see \Cref{lem:dim-bound}. It also has a natural convex-analytic interpretation: by \Cref{app:lem:nuclear-hermite}, $\tr(\mathbf{M}(P)^{1/2})$ is the nuclear norm (\Cref{app:nuclear-def}) of the Hermite coefficient matrix of $\nabla P$.
Overall, the trace of $\mathbf{M}(P)^{1/2}$ can be viewed as a convex surrogate for controlling the rank of $\mathbf{M}(P)$, in the same way that the lasso is used as a convex surrogate for sparsity \cite{Tibshirani1996, Donoho2006}.

Following this approach yields the single-halfspace learner given below.
\begin{algorithm}[H]
\caption{\protect\algoname{1}}
\label{alg:proper-halfspaces}
\begin{algorithmic}[1]
\Procedure{Agnostic-Proper-Learner}{$\eps,\delta,D$}
\State \textbf{Input:} $\eps>0$, $\delta>0$ and sample access to a distribution $D$ on $\R^d\times\{\pm1\}$
\State \textbf{Output:} A halfspace $h$ such that $\err_D(h)\le \OPT+\eps$ with probability at least $1-\delta$
\State Set $k \gets C_0 \eps^{-2}$ and $\eta \gets \eps^2/C_0$ \hfill $\triangleright$ $C_0$ is a sufficiently large absolute const
\State Set $\mu\gets 1/128$ and $\nu\gets c_\nu\eps^{3/2}$ \hfill $\triangleright$ $c_\nu$ is a sufficiently small absolute const
\State Find $P\in\cP_k$ such that $\cL_{\mu,\nu}(P)\le \inf_{P'\in\cP_k} \cL_{\mu,\nu}(P')+O(\eps^3)$ using \Cref{prop:compute-P}
\State Let $\mathbf{M}(P)=\E[\nabla P(\mathbf{x})\nabla P(\mathbf{x})^\top]$
\State Let $V$ be the span of the eigenvectors of $\mathbf{M}(P)$ whose eigenvalues are at least $\eta$
\State Construct a  $\eps$-cover $\cH$ of halfspaces with normals in $V$ using 
\Cref{fact:cover}
\State Draw $N_2=O\!\left(\log(|\cH|/\delta)/\eps^2\right)$ i.i.d. samples from $D$ and let $\widehat D$ be the empirical distribution
\State $h \gets \arg\min_{h'\in\cH}\err_{\widehat D}(h')$
\State \Return $h$
\EndProcedure
\end{algorithmic}
\end{algorithm}

\subsection{Extension to functions of halfspaces}
\label{sec:extension-functions-halfspaces}

We show how to extend the approach above from a single halfspace to arbitrary Boolean functions of several halfspaces. The core extension is based on the same principle as in the halfspace case: a polynomial is used as a device for identifying the low-dimensional structure on which a near-optimal classifier depends. Once this structure is recovered, we discard the polynomial and perform a search inside the original concept class.

The main additional difficulty is that the optimal classifier may now combine several halfspaces through an arbitrary Boolean rule. Thus, the decision boundary need not be close to a single hyperplane. Nevertheless, such classifiers still depend on the input only through the normals of the underlying halfspaces. Consequently, any near-optimal function of halfspaces varies meaningfully only along a low-dimensional subspace generated by these normals.

Our regularized polynomial objective \eqref{logistic-loss} is designed to recover such subspace. As before, we optimize a convex surrogate loss over low-degree polynomials (\Cref{prop:compute-P}) and then examine the gradient matrix $\mathbf{M}(P)=\E[\nabla P(\mathbf{x})\nabla P(\mathbf{x})^\top]$ of the computed polynomial. The spectral regularizer prevents the polynomial from spreading its variation across too many directions, so the large-eigenvalue eigenspace captures the directions that are relevant for the near-optimal function of $K$ halfspaces. The structural part of the argument is formalized in \Cref{thm:structural-B}, which shows that this large-eigenvalue subspace contains the directions needed to approximate a near-optimal function of halfspaces.

The main technical difference from the case of proper agnostic learning of a single halfspace is in the analysis --- we can no longer average directions one at a time. Instead, we need to average over an entire subspace of low-influence directions. This is where the subspace Gaussian Poincaré inequality (\Cref{app:subspace_gp}) is used. This smoothing step is formalized in \Cref{lem:subspace-poincare} and is the key ingredient that lets the dimension-reduction argument extend from a single halfspace to Boolean functions of halfspaces. 

Once this subspace is identified, we construct a  $\eps$-cover of functions of $K$ halfspaces restricted to this subspace and search for the best candidate in the cover. The covering argument, stated in \Cref{fact:cover}, shows that this cover contains a proper hypothesis whose error is close to the best achievable error among all Boolean functions of $K$ halfspaces. Finally, empirical risk minimization over the cover selects such a hypothesis with high probability. \Cref{lem:boolean-map-erm} and \Cref{fact:cover} together show how to perform this optimization step efficiently. Combining these ingredients gives the proper agnostic learning guarantee in \Cref{thm:main-B}.
\textbf{\begin{algorithm}[H]
\caption{Agnostic Proper Learner for Boolean functions of Halfspaces}
\label{alg:proper-boolean}
\begin{algorithmic}[1]
\Procedure{Agnostic-Proper-Learner-Boolean-Functions}{$K,\eps,\delta,D$}
\State \textbf{Input:} $K\ge 2$, $\eps>0$, $\delta>0$ and sample access to a distribution $D$ on $\R^d\times\{\pm1\}$
\State \textbf{Output:} A hypothesis $h\in\mathcal{C}_K$ such that $\err_D(h)\le \OPT_K+\eps$ with probability at \\
\hspace*{3em} least $1-\delta$
\State Set $k \gets C_0 K^2\log(1/\eps)/\eps^2$, $\eta \gets \eps^2/(C_0K)$ \hfill $\triangleright$ $C_0$ is a sufficiently large absolute const
\State Set $\mu\gets 1/128$, $\nu\gets c_\nu\eps^{3/2}/K^{3/2}$ \hfill $\triangleright$ $c_\nu$ is a sufficiently small absolute const
\State Find $P\in\cP_k$ such that $\cL_{\mu,\nu}(P)\le \inf_{P'\in\cP_k} \cL_{\mu,\nu}(P')+O(\eps^3)$ using \Cref{prop:compute-P} 
\State Let $\mathbf{M}(P)=\E[\nabla P(\mathbf{x})\nabla P(\mathbf{x})^\top]$
\State Let $\mathbf{V}$ be the span of the eigenvectors of $M(P)$ whose eigenvalues are at least $\eta$
\State Construct a $c\eps/K$-cover $\cH$ of halfspaces with normals in $\mathbf{V}$ using 
\Cref{fact:cover}
\State Draw
$
N_2=O\!\left((K\log|\cH|+2^K+\log(1/\delta))/\eps^2\right)
$
i.i.d. samples from $D$ and let $\widehat D$ be\\
\hspace*{3em} the empirical distribution
\State For each tuple $(h_1,\dots,h_K)\in\cH^K$, compute the empirically optimal Boolean map \\
\hspace*{3em} $B:\{\pm1\}^K\to\{\pm1\}$ using \Cref{lem:boolean-map-erm}
\State Return
$
h \gets \arg\min_{\substack{h_1,\dots,h_K\in\cH\\ B:\{\pm1\}^K\to\{\pm1\}}}
\err_{\widehat D}\!\big(\mathbf{x}\mapsto B(h_1(\mathbf{x}),\dots,h_K(\mathbf{x}))\big)
$
\EndProcedure
\end{algorithmic}
\end{algorithm}}
\subsection{Related Work}

\paragraph{Proper Learning.}

Improperness is a well-known limitation of the agnostic-learning-from-approximation paradigm of \cite{KKMS08}, and there has been significant interest in designing proper agnostic learning algorithms matching its run-time. Despite this interest, few such algorithms are known: in addition to this work and the already-discussed \cite{DiakonikolasKaneKontonisTzamosZarifis2021}, algorithms exist for proper agnostic learning of monotone functions \cite{lange2022properly, lange2025agnostic} and convex functions \cite{pinto2025learning}. The algorithm of \cite{mehta2002decision} can also match the run-time of \cite{KKMS08} for proper agnostic learning of decision trees from uniform samples. Our work adds intersections of halfspaces and Boolean functions of halfspaces to this list, while also improving the run-time for proper agnostic learning of halfspaces.

In the \emph{weakly agnostic} setting, where the goal is a classifier with error $O(\OPT)+\eps$ rather than $\OPT+\eps$, the state-of-the-art algorithms \cite{awasthi2017power,diakonikolas2018learning,diakonikolas2022learning, diakonikolas2020non} for learning halfspaces under the Gaussian distribution are already proper. %Similarly, \cite{daniely2015ptas} gave a poly-time improper weakly-agnostic algorithm with error $(1+\mu)\OPT+\eps$ for arbitrarily small absolute constant $\mu$, and \cite{DiakonikolasKaneKontonisTzamosZarifis2021} matched its run-time and accuracy with a proper hypothesis. 
These weakly-agnostic results are specific to halfspaces and do not extend to intersections or functions of halfspaces.

In the \emph{realizable}\footnote{The work of \cite{diakonikolas2018learning} also handles \emph{nasty noise}, with error polynomial in the noise rate.} setting --- where labels are described exactly by a function in the class --- \emph{partially proper} algorithms are known for learning intersections of halfspaces over the Gaussian distribution \cite{vempala2010learningconvex,vempala2010randomsampling,vempala2010corrigendum,diakonikolas2018learning}: when learning an intersection of $k$ halfspaces, they return an intersection of $k'>k$ halfspaces. Prior to our work, no efficient algorithm was known for learning the class of intersections of $k$ halfspaces that gives a hypothesis that  is itself an intersection of $k$ halfspaces (even in the realizable setting). For Boolean functions of halfspaces, even partially proper efficient algorithms were unknown.

Additionally, in the realizable setting, the recent work of \cite{alman2026learning} gives an improper learning algorithm for intersections of halfspaces under arbitrary data distributions with run-time $2^{\sqrt{d}(\log d)^{O(k)}} \poly(1/\eps)$.

So far we have discussed only proper learning algorithms with access to i.i.d.\ labeled samples. Proper agnostic learning has also been studied with queries; see \cite{ehrenfeucht1989learning, blanc2022properly} for such work on decision trees and \cite{diakonikolas2024agnostically} for work on halfspaces and RELUs.

\paragraph{Convex surrogates.} Convex surrogate losses have long been central in statistics and machine learning, beginning with logistic models \cite{Berkson1944, Cox1958} and later support vector machines and boosting \cite{CortesVapnik1995, FriedmanHastieTibshirani2000}. Although the $0$-$1$ loss is hard to optimize directly due to discontinuity and nonconvexity, convex surrogates are more than computational relaxations: they yield a margin-based formulation in which excess surrogate risk can be related to excess classification risk \cite{BartlettJordanMcAuliffe2006}.

\paragraph{Regularization.} Regularization was developed to address instability in statistical estimation, particularly for shrinkage and ill-posed problems \cite{HoerlKennard1970,TikhonovArsenin1977}. Norm penalties control function-class size, and thus margin, complexity, and generalization \cite{CortesVapnik1995,BartlettMendelson2002,ScholkopfSmola2002}. Different regularizers correspond to different structure: $L_2$ controls norm, $L_1$ promotes sparsity, and nuclear-norm regularization captures low-rank structure \cite{Tibshirani1996,SrebroRennieJaakkola2004,RechtFazelParrilo2010}.

\section{Notation}
\label{app:notation}
We work with distributions $D$ over pairs $(\mathbf{x},y)\in\R^d\times\{\pm1\}$, where $\mathbf{x}\sim\cN(0,\mathbf{I}_d)$. The label $y$ may depend on $\mathbf{x}$ arbitrarily. For a classifier $h$, write $\err_D(h):=\Pr_{(\mathbf{x},y)\sim D}[h(\mathbf{x})\neq y]$.

For the pointwise loss we use $\ell$, and for the corresponding population objective we use $\cL$. With regularization, we write $\cL_{\mu,\nu}$, where $\mu>0$ and $\nu\ge0$. We reserve $\lambda$ for eigenvalues, and $\|\cdot\|_*$ denotes the nuclear norm.

Let $\cP_k$ be the space of polynomials on $\R^d$ of degree at most $k$. Let $\mathcal{C}$ be the class of halfspaces $\sign(\langle \mathbf{w},\mathbf{x}\rangle+t)$ with $\|\mathbf{w}\|_2=1$, together with the two constant classifiers. If $V\subseteq\R^d$ is a subspace, $\mathcal{C}_V$ denotes the subclass whose normal vectors lie in $V$, again including the two constant classifiers. For $K\ge1$, let $\mathcal{C}_K$ be the class of arbitrary Boolean functions of $K$ halfspaces, $\mathcal{C}_K := \{\mathbf{x}\mapsto B(f_1(\mathbf{x}),\dots,f_K(\mathbf{x})): f_1,\dots,f_K\in\mathcal{C},\ B:\{-1,+1\}^K\to\{-1,+1\}\}$. Similarly, $\mathcal{C}_{K,V}$ denotes the subclass in which all halfspace normal vectors lie in $V$. We write $\OPT:=\inf_{f\in\mathcal{C}}\err_D(f)$, $\OPT_K:=\inf_{f\in\mathcal{C}_K}\err_D(f)$, and more generally $\OPT_{\mathcal F}:=\inf_{f\in\mathcal F}\err_D(f)$.

For a subspace $V\subseteq\R^d$, write $V^\perp$ for its orthogonal complement and $\mathbf{x}_V$ for the orthogonal projection of $\mathbf{x}$ onto $V$. We use $O(\cdot)$, $\Theta(\cdot)$, and $o(\cdot)$ with absolute hidden constants unless stated otherwise, and write $\poly(\cdot)$ for a polynomial in its arguments.

\section{Proper Learning of Halfspaces}\label{halfspace}

The single-halfspace learner is restated below.
\begin{algorithm}[H]
\caption*{Algorithm \ref{alg:proper-halfspaces} (restated): \protect\algoname{1}}
\begin{algorithmic}[1]
\Procedure{Agnostic-Proper-Learner}{$\eps,\delta,D$}
\State \textbf{Input:} $\eps>0$, $\delta>0$ and sample access to a distribution $D$ on $\R^d\times\{\pm1\}$
\State \textbf{Output:} A halfspace $h$ such that $\err_D(h)\le \OPT+\eps$ with probability at least $1-\delta$
\State Set $k \gets C_0 \eps^{-2}$ and $\eta \gets \eps^2/C_0$ \hfill $\triangleright$ $C_0$ is a sufficiently large absolute const
\State Set $\mu\gets 1/128$ and $\nu\gets c_\nu\eps^{3/2}$ \hfill $\triangleright$ $c_\nu$ is a sufficiently small absolute const
\State Find $P\in\cP_k$ such that $\cL_{\mu,\nu}(P)\le \inf_{P'\in\cP_k} \cL_{\mu,\nu}(P')+O(\eps^3)$ using \Cref{prop:compute-P}
\State Let $\mathbf{M}(P)=\E[\nabla P(\mathbf{x})\nabla P(\mathbf{x})^\top]$
\State Let $V$ be the span of the eigenvectors of $\mathbf{M}(P)$ whose eigenvalues are at least $\eta$
\State Construct a  $\eps$-cover $\cH$ of halfspaces with normals in $V$ using 
\Cref{fact:cover}
\State Draw $N_2=O\!\left(\log(|\cH|/\delta)/\eps^2\right)$ i.i.d. samples from $D$ and let $\widehat D$ be the empirical distribution
\State $h \gets \arg\min_{h'\in\cH}\err_{\widehat D}(h')$
\State \Return $h$
\EndProcedure
\end{algorithmic}
\end{algorithm}

We will also use the following standard facts.
To compute a polynomial satisfying the required near-optimality condition for this objective, we use the following result based on a standard uniform convergence argument.
\begin{proposition}[Regularized logistic-loss polynomial regression]\label{app:prop:compute-P}
Fix $\mu>0$ to be an absolute constant and let $\nu\ge 0$. Let $D$ be a distribution over $(\mathbf{x},y)\in\R^d\times\{\pm1\}$ whose $\mathbf{x}$-marginal is $\cN(0,\mathbf{I}_d)$. Let $k\in\mathbb Z_+$ and $\eps,\delta>0$. There is an algorithm that draws
$
N=(dk)^{O(k)}\frac{\log(1/\delta)}{\eps^2}
$
samples from $D$, runs in time $\poly(N,d)$, and outputs a polynomial $P\in\cP_k$ such that
$
\cL_{\mu,\nu}(P)\le \inf_{P'\in\cP_k}\cL_{\mu,\nu}(P')+\eps
$
with probability at least $1-\delta$.
\end{proposition}

A proof can be found in Appendix~\ref{app:compute_reg}. We remind our reader that this polynomial $P$ is used to construct the matrix $\mathbf{M}(P)=\E[\nabla P(\mathbf{x})\nabla P(\mathbf{x})^\top]$.

We will also use the fact that linear threshold functions can be approximated well by low-degree polynomials. 
\begin{theorem}[Gaussian $L_1$ approximation]\label{app:thm:threshold-approx}
There exists an absolute constant $C_{\mathrm{HS}}$ with the following property. For every $\tau\in(0,1/10)$ and every shift $b\in\R$, there exists a univariate polynomial $S:\R\to\R$ of degree at most $C_{\mathrm{HS}}\tau^{-2}$ such that
$$
\E_{x\sim\cN(0,1)}[\,|S(x)-\sign(x+b)|\,]\le \tau,
$$
$$
\E_{x\sim\cN(0,1)}[S'(x)^2]\le C\tau^{-1},
$$
$$
\E_{x\sim\cN(0,1)}[S(x)^2]\le 5.
$$

\end{theorem}

A proof is included in Appendix~\ref{app:approx}. It is based on the one-dimensional FT-mollification construction of \cite{DKN10}; the only additional requirement we verify is that, besides $L^1$-approximation, we also require the control of $L^2$ and gradient norms. 

The lemma below shows how the label-flip symmetry can be leveraged for decomposing the loss into a label-dependent and a model-dependent component.
\begin{lemma}\label{lem:label-flip-decomposition}
Assume $\psi:\R\to\R$ satisfies $\psi(t)+\psi(-t)=1$ for all $t\in\R$, and set
$\phi(u):=\frac{1}{2}-\psi(u)$. Then $\phi$ is odd and, for all $y\in\{\pm1\}$ and $u\in\R$,
\begin{equation}\label{eq:label-flip-symmetry}
y\,\psi(yu)=\frac{1}{2}\,y-\phi(u).
\end{equation}
If $\psi$ is $L$-Lipschitz, then so is $\phi$.
\end{lemma}

\begin{proof}
The symmetry gives
$\phi(-u)=\frac{1}{2}-\psi(-u)=\frac{1}{2}-(1-\psi(u))=-\phi(u)$,
so $\phi$ is odd. Hence, since $y\in\{\pm1\}$,
$
y\,\psi(yu)
=
y\left(\frac{1}{2}-\phi(yu)\right)
=
\frac{1}{2}\,y-y\phi(yu)
=
\frac{1}{2}\,y-\phi(u).
$
Finally, $|\phi(u)-\phi(v)|=|\psi(u)-\psi(v)|\le L|u-v|$.
\end{proof}

The next lemma is the conditional version of Gaussian Poincar\'e inequality.

\begin{lemma}[Conditional Gaussian Poincar\'e]\label{lem:cond-poincare}
Let $\boldsymbol{\xi}\in\R^d$ be a unit vector and write $\mathbf{x} = \mathbf{u} + z\boldsymbol{\xi}$ with $\mathbf{u}\in \boldsymbol{\xi}^\perp$ and $z\sim\cN(0,1)$ independent. For every weakly differentiable $G:\R^d\to\R$,
$$
\E\!\left[\Big(G(\mathbf{x})-\E[G(\mathbf{x})\mid \mathbf{u}]\Big)^2\right]
\le \E\big[(\partial_{\boldsymbol{\xi}} G(\mathbf{x}))^2\big].
$$
\end{lemma}

A proof can be found in Appendix~\ref{app:gp}. 

\subsection{Reducing to a Low-Dimensional Search Space}\label{sec:structural}

Let $V\subseteq \R^d$ be a subspace and let
$
f(\mathbf{x})=\sign(\langle \mathbf{w},\mathbf{x}\rangle+t)
$
be a halfspace with $\|\mathbf{w}\|_2=1$. Decompose $\mathbf{w}=\mathbf{w}_{V}+\mathbf{w}_{V^\perp}$, where $\mathbf{w}_{V}$ is in $V$ and $\mathbf{w}_{V^\perp}$ is orthogonal to $V$. If $\mathbf{w}_{V^\perp}\neq 0$, let
$
\boldsymbol{\xi}:=\frac{\mathbf{w}_{V^\perp}}{\|\mathbf{w}_{V^\perp}\|_2}\in V^\perp
$
be the normalized missing direction. Decompose $\mathbf{x}=\mathbf{x}_{\boldsymbol{\xi}^\perp}+\mathbf{x}_{\boldsymbol{\xi}}$, where $\mathbf{x}_{\boldsymbol{\xi}^\perp}$ and $\mathbf{x}_{\boldsymbol{\xi}}$ are the projections of $\mathbf{x}$ onto $\boldsymbol{\xi}^\perp$ and $\operatorname{span}\{\boldsymbol{\xi}\}$, respectively. For $z\sim\cN(0,1)$ independent of $\mathbf{x}_{\boldsymbol{\xi}^\perp}$, we have a resampling identity
\begin{equation}\label{eq:resample}
\mathbf{x}=\mathbf{x}_{\boldsymbol{\xi}^\perp}+\mathbf{x}_{\boldsymbol{\xi}} \overset{\mathrm d}{=} \mathbf{x}_{\boldsymbol{\xi}^\perp}+z\boldsymbol{\xi}.
\end{equation}

\begin{definition}[Averaging a halfspace]\label{app:def:halfspace-averaging}
For each $z\in\R$, define the shifted halfspace
$
h_z(\mathbf{x}):=f(\mathbf{x}_{\boldsymbol{\xi}^\perp}+z\boldsymbol{\xi}).
$
Based on the resampling identity\;\eqref{eq:resample}, define
\begin{equation}\label{app:eq:averaged-halfspace}
f_{V}(\mathbf{x}):=\E_{z\sim\cN(0,1)}[h_z(\mathbf{x})]
=
\E_{z\sim\cN(0,1)}[f(\mathbf{x}_{\boldsymbol{\xi}^\perp}+z\boldsymbol{\xi})].
\end{equation}
If $\mathbf{w}_{V^\perp}=0$, define $f_{V}(\mathbf{x}):=f(\mathbf{x})$.
\end{definition}

For each fixed $z\in\R$, the shifted function $h_z$ belongs to $\mathcal{C}_{V}$ (i.e. the concept class of halfspaces whose normals are in $V$). Indeed, set
$
t_z:=t+z\langle \mathbf{w},\boldsymbol{\xi}\rangle.
$
Since $\boldsymbol{\xi}\in V^\perp$ and $\mathbf{w}_{V^\perp}$ is parallel to $\boldsymbol{\xi}$, we have
\begin{equation}\label{app:eq:averaged-halfspace-shift}
h_z(\mathbf{x})
=
f(\mathbf{x}_{\boldsymbol{\xi}^\perp}+z\boldsymbol{\xi})
=
\sign(\langle \mathbf{w}_{V},\mathbf{x}\rangle+t_z).
\end{equation}
Thus each shifted halfspace has normal vector in $V$. Therefore $h_z\in\mathcal{C}_{V}$ for every fixed $z$, and consequently $f_{V}\in\conv(\mathcal{C}_{V})$.

\begin{remark}
If $\mathbf{w}_{V}=0$, the corresponding shifted halfspace is constant, which we regard as a degenerate halfspace.
\end{remark}

Using \Cref{app:def:halfspace-averaging} and the definition of the regularized logistic loss $\eqref{logistic-loss}$, we can construct a subspace that contains a near-optimal halfspace.

\begin{theorem}[Existence of a Low-Dimensional Subspace for Halfspaces]\label{app:thm:structural}
There exists an absolute constant $C>0$ and $c_\nu>0$ such that the following holds. Fix $\eps\in(0,1/10)$, set
$k:=C/\eps^2$, $\mu:=1/128$, $\eta:=\eps^2/1024$, and $\nu:=c_\nu\eps^{3/2}$. Let $P\in\cP_k$ satisfy $\cL_{\mu,\nu}(P)\le \inf_{P'\in\cP_k}\cL_{\mu,\nu}(P')+O(\eps^3)$,
where the hidden constant is sufficiently small. Let $V$ be the span of the eigenvectors of $\mathbf{M}(P)$ whose eigenvalues are at least $\eta$. Then for every halfspace $f\in\mathcal{C}$, the averaged function $f_V$ from \Cref{app:def:halfspace-averaging} satisfies
\begin{equation}\label{eq:averaging-main-bound}
\E_{(\mathbf{x},y)\sim D}[(f(\mathbf{x})-f_V(\mathbf{x}))y]\le \eps.
\end{equation}
Consequently, there exists $h\in\mathcal{C}_V$ such that
$$
\E_{(\mathbf{x},y)\sim D}[(f(\mathbf{x})-h(\mathbf{x}))y]\le \eps.
$$
\end{theorem}

The idea of the proof is to show that averaging the halfspace over the missing directions does not lose much correlation with the labels. The quantity $f-f_V$ measures the contribution of the directions outside $V$ that is removed by averaging. Thus, proving that $\E[(f-f_V)y]$ is small means that these missing directions do not carry significant label correlation.

The proof has four main steps. First, \Cref{lem:logistic-gradient-correlation} lower bounds the correlation of $f-f_V$ with the weighted labels $y\psi(yP)$ in terms of its correlation with the labels. Second, \Cref{lem:q-approx} approximates $f-f_V$ by a low-degree polynomial $Q$. Third, \Cref{lem:regularizer-term} shows that the regularization cost of this perturbation is small. Finally, \Cref{lem:no-large-gain} uses approximate optimality of $P$ to upper bound the possible gain from such perturbations.

\begin{proof}
Fix a halfspace $f(\mathbf{x})=\sign(\langle \mathbf{w},\mathbf{x}\rangle+t)$. If $\mathbf{w}_{V^\perp}=0$, then $f_V(\mathbf{x})=f(\mathbf{x})$, and the claim is immediate. Hence assume $\mathbf{w}_{V^\perp}\neq 0$ and let $\boldsymbol{\xi}=\mathbf{w}_{V^\perp}/\|\mathbf{w}_{V^\perp}\|_2$. 

The next lemma controls the correlation between $f-f_V$ and the negative derivative of the logistic loss, $\psi(yP(\mathbf{x}))$. We remind the reader that the function $\psi$ is defined in \eqref{logistic-loss}.

\begin{lemma}[Correlation with the logistic gradient]\label{lem:logistic-gradient-correlation}
For the averaged function $f_V$ from \Cref{app:def:halfspace-averaging},
\begin{equation}\label{eq:g-logistic-grad}
\E_{(\mathbf{x},y)\sim D}[(f(\mathbf{x})-f_V(\mathbf{x}))\,y\,\psi(yP(\mathbf{x}))]
\ge
\frac{1}{2}\E_{(\mathbf{x},y)\sim D}[(f(\mathbf{x})-f_V(\mathbf{x}))y]-\frac{1}{2}\sqrt{\eta}.
\end{equation}
\end{lemma}
\begin{proof}
Using the label-flip property\;\eqref{eq:label-flip-symmetry},
\begin{equation}\label{eq:symmetric-corollary}
\E_{(\mathbf{x},y)\sim D}[(f(\mathbf{x})-f_V(\mathbf{x}))\,y\,\psi(yP(\mathbf{x}))]
=
\frac{1}{2}\E_{(\mathbf{x},y)\sim D}[(f(\mathbf{x})-f_V(\mathbf{x}))y]
-
\E_{\mathbf{x}\sim\cN(0,\mathbf{I}_d)}[(f(\mathbf{x})-f_V(\mathbf{x}))\phi(P(\mathbf{x}))].
\end{equation}
It remains to control the second term in Eq.~\eqref{eq:symmetric-corollary}. The needed estimate is
\begin{equation}\label{eq:gphi-bound}
\E_{\mathbf{x}\sim\cN(0,\mathbf{I}_d)}[(f(\mathbf{x})-f_V(\mathbf{x}))\phi(P(\mathbf{x}))]\le \frac{1}{2}\sqrt{\eta}.
\end{equation}
Let $G(\mathbf{x}):=\phi(P(\mathbf{x}))$ and $G_V(\mathbf{x}):=\E_{z\sim\cN(0,1)}[G(\mathbf{x}_{\boldsymbol{\xi}^\perp}+z\boldsymbol{\xi})]$. Since $G_V(\mathbf{x})$ depends only on $\mathbf{x}_{\boldsymbol{\xi}^\perp}$,
$$
\E_{\mathbf{x}\sim\cN(0,\mathbf{I}_d)}[(f(\mathbf{x})-f_V(\mathbf{x}))G_V(\mathbf{x})]
=
\E_{\mathbf{x}\sim\cN(0,\mathbf{I}_d)}\!\left[G_V(\mathbf{x})\,\E_{\mathbf{x}\sim\cN(0,\mathbf{I}_d)}[f(\mathbf{x})-f_V(\mathbf{x})\mid \mathbf{x}_{\boldsymbol{\xi}^\perp}]\right]
=
0,
$$
where the last equality follows from the definition of $f_V$. Hence
$
\E_{\mathbf{x}\sim\cN(0,\mathbf{I}_d)}[(f(\mathbf{x})-f_V(\mathbf{x}))\phi(P(\mathbf{x}))]
=
\E_{\mathbf{x}\sim\cN(0,\mathbf{I}_d)}[(f(\mathbf{x})-f_V(\mathbf{x}))(G(\mathbf{x})-G_V(\mathbf{x}))].
$
By \Cref{lem:cond-poincare},
$
\E_{\mathbf{x}\sim\cN(0,\mathbf{I}_d)}[(G(\mathbf{x})-G_V(\mathbf{x}))^2]\le \E_{\mathbf{x}\sim\cN(0,\mathbf{I}_d)}[(\partial_{\boldsymbol{\xi}} G(\mathbf{x}))^2].
$
Since $|\phi'|\le 1/4$,
$$
\E_{\mathbf{x}\sim\cN(0,\mathbf{I}_d)}[(\partial_{\boldsymbol{\xi}} G(\mathbf{x}))^2]
\le
\frac{1}{16}\E_{\mathbf{x}\sim\cN(0,\mathbf{I}_d)}[(\partial_{\boldsymbol{\xi}} P(\mathbf{x}))^2]
=
\frac{1}{16}\boldsymbol{\xi}^\top \mathbf{M}(P)\boldsymbol{\xi}
\le
\frac{\eta}{16},
$$
because $\boldsymbol{\xi}\in V^\perp$ and all eigenvalues of $\mathbf{M}(P)$ on $V^\perp$ are smaller than $\eta$. Finally, $|f(\mathbf{x})-f_V(\mathbf{x})|\le 2$, so the Cauchy--Schwarz inequality gives
$$
\E_{\mathbf{x}\sim\cN(0,\mathbf{I}_d)}[(f(\mathbf{x})-f_V(\mathbf{x}))(G(\mathbf{x})-G_V(\mathbf{x}))] \le \|f-f_V\|_2\|G-G_V\|_2 \le 2\cdot\frac{\sqrt{\eta}}{4} = \frac{1}{2}\sqrt{\eta}.
$$
This proves the Eq.\;\eqref{eq:gphi-bound}. Combining Eq.\;\eqref{eq:symmetric-corollary} and \eqref{eq:gphi-bound} gives the bound\;\eqref{eq:g-logistic-grad}.
\end{proof}

\begin{lemma}[Low-degree approximation of $f-f_V$]\label{lem:q-approx}
Fix $\tau\in(0,1/10)$. Suppose $k\ge C_{\mathrm{HS}}\tau^{-2}$. Then there exists $Q\in\cP_k$ such that $\E_{\mathbf{x}\sim\cN(0,\mathbf{I}_d)}[Q(\mathbf{x})^2]\le 5$ and
\begin{equation}\label{eq:q-approx-bound}
\E_{\mathbf{x}\sim\cN(0,\mathbf{I}_d)}[|Q(\mathbf{x})-(f(\mathbf{x})-f_V(\mathbf{x}))|]\le 2\tau.
\end{equation}
\end{lemma}

\begin{proof}
By \Cref{app:thm:threshold-approx}, there exists a polynomial $S$ of degree at most $C_{\mathrm{HS}}\tau^{-2}$ such that $\E_{\mathbf{x}\sim\cN(0,\mathbf{I}_d)}[|S(\mathbf{x})-f(\mathbf{x})|]\le \tau$ and $\E_{\mathbf{x}\sim\cN(0,\mathbf{I}_d)}[S(\mathbf{x})^2]\le 5$. Since $k\ge C_{\mathrm{HS}}\tau^{-2}$, we have $S\in\cP_k$. Define
$
S_V(\mathbf{x}):=\E_{z\sim\cN(0,1)}[S(\mathbf{x}_{\boldsymbol{\xi}^\perp}+z\boldsymbol{\xi})]
$
and set $Q(\mathbf{x}):=S(\mathbf{x})-S_V(\mathbf{x})$. Gaussian averaging in the $\boldsymbol{\xi}$-direction preserves degree, so $Q\in\cP_k$. Moreover,
$$
\E_{\mathbf{x}\sim\cN(0,\mathbf{I}_d)}[Q(\mathbf{x})^2]
=
\E_{\mathbf{x}\sim\cN(0,\mathbf{I}_d)}[\Var(S(\mathbf{x})\mid \mathbf{x}_{\boldsymbol{\xi}^\perp})]
\le
\E_{\mathbf{x}\sim\cN(0,\mathbf{I}_d)}[S(\mathbf{x})^2]
\le
5.
$$
Also, $Q(\mathbf{x})-(f(\mathbf{x})-f_V(\mathbf{x}))=(S(\mathbf{x})-f(\mathbf{x}))-(S_V(\mathbf{x})-f_V(\mathbf{x}))$. Hence
$$
\E_{\mathbf{x}\sim\cN(0,\mathbf{I}_d)}[|Q(\mathbf{x})-(f(\mathbf{x})-f_V(\mathbf{x}))|]
\le
\E_{\mathbf{x}\sim\cN(0,\mathbf{I}_d)}[|S(\mathbf{x})-f(\mathbf{x})|]
+
\E_{\mathbf{x}\sim\cN(0,\mathbf{I}_d)}[|S_V(\mathbf{x})-f_V(\mathbf{x})|].
$$
By Jensen's inequality,
$$
\E_{\mathbf{x}\sim\cN(0,\mathbf{I}_d)}[|S_V(\mathbf{x})-f_V(\mathbf{x})|]
\le
\E_{\mathbf{x}\sim\cN(0,\mathbf{I}_d)}\!\left[
\E_{z\sim\cN(0,1)}[|S(\mathbf{x}_{\boldsymbol{\xi}^\perp}+z\boldsymbol{\xi})-f(\mathbf{x}_{\boldsymbol{\xi}^\perp}+z\boldsymbol{\xi})|]
\right]
=
\E_{\mathbf{x}\sim\cN(0,\mathbf{I}_d)}[|S(\mathbf{x})-f(\mathbf{x})|],
$$
where the last equality uses the Gaussian resampling identity $\mathbf{x}\overset{d}=\mathbf{x}_{\boldsymbol{\xi}^\perp}+z\boldsymbol{\xi}$. Therefore $\E_{\mathbf{x}\sim\cN(0,\mathbf{I}_d)}[|Q(\mathbf{x})-(f(\mathbf{x})-f_V(\mathbf{x}))|]\le 2\tau$.
\end{proof}

\begin{lemma}[Control of the regularization terms]\label{lem:regularizer-term}
Let $Q=S-S_V$ be the polynomial constructed in \Cref{lem:q-approx}. Then
\begin{equation}\label{eq:PQ-bound}
\E_{\mathbf{x}\sim\cN(0,\mathbf{I}_d)}[P(\mathbf{x})Q(\mathbf{x})]\le \sqrt{5\eta}.
\end{equation}
\begin{equation}\label{eq:Q-nuclear-bound}
\tr(\mathbf{M}(Q)^{1/2})\le O(\tau^{-1/2}).
\end{equation}
\end{lemma}

\begin{proof}
Recall that $Q(\mathbf{x})=S(\mathbf{x})-S_V(\mathbf{x})$, where $S_V(\mathbf{x})=\E_{z\sim\cN(0,1)}[S(\mathbf{x}_{\boldsymbol{\xi}^\perp}+z\boldsymbol{\xi})]$. Hence $\E_{\mathbf{x}\sim\cN(0,\mathbf{I}_d)}[Q(\mathbf{x})\mid \mathbf{x}_{\boldsymbol{\xi}^\perp}]=0$. Let $P_V(\mathbf{x}):=\E_{z\sim\cN(0,1)}[P(\mathbf{x}_{\boldsymbol{\xi}^\perp}+z\boldsymbol{\xi})]$. Since $P_V(\mathbf{x})$ depends only on $\mathbf{x}_{\boldsymbol{\xi}^\perp}$, we have $\E_{\mathbf{x}\sim\cN(0,\mathbf{I}_d)}[P_V(\mathbf{x})Q(\mathbf{x})]=0$, and therefore
$
\E_{\mathbf{x}\sim\cN(0,\mathbf{I}_d)}[P(\mathbf{x})Q(\mathbf{x})]
=
\E_{\mathbf{x}\sim\cN(0,\mathbf{I}_d)}[(P(\mathbf{x})-P_V(\mathbf{x}))Q(\mathbf{x})].
$
By \Cref{lem:cond-poincare},
\begin{equation}\label{eq:polynomial-poincare}
\E_{\mathbf{x}\sim\cN(0,\mathbf{I}_d)}[(P(\mathbf{x})-P_V(\mathbf{x}))^2]
\le
\E_{\mathbf{x}\sim\cN(0,\mathbf{I}_d)}[(\partial_{\boldsymbol{\xi}} P(\mathbf{x}))^2]
\le
\eta.
\end{equation}
Using Eq.\;\eqref{eq:polynomial-poincare} and $\E_{\mathbf{x}\sim\cN(0,\mathbf{I}_d)}[Q(\mathbf{x})^2]\le 5$ gives
$$
\E_{\mathbf{x}\sim\cN(0,\mathbf{I}_d)}[P(\mathbf{x})Q(\mathbf{x})]
\le
\Big(
\E_{\mathbf{x}\sim\cN(0,\mathbf{I}_d)}
\big[(P(\mathbf{x})-P_V(\mathbf{x}))^2\big]\,
\E_{\mathbf{x}\sim\cN(0,\mathbf{I}_d)}
\big[Q(\mathbf{x})^2\big]
\Big)^{1/2} 
\le \sqrt{5\eta}.
$$
This proves Eq.\;\eqref{eq:PQ-bound}.

We now prove Eq.\;\eqref{eq:Q-nuclear-bound}. By \Cref{app:thm:threshold-approx}, $S$ also satisfies $\E_{\mathbf{x}\sim\cN(0,\mathbf{I}_d)}\|\nabla S(\mathbf{x})\|^2\le O(\tau^{-1})$. Since $S_V(\mathbf{x})=\E_{z\sim\cN(0,1)}[S(\mathbf{x}_{\boldsymbol{\xi}^\perp}+z\boldsymbol{\xi})]$, Jensen's inequality gives $\E_{\mathbf{x}\sim\cN(0,\mathbf{I}_d)}\|\nabla S_V(\mathbf{x})\|^2\le \E_{\mathbf{x}\sim\cN(0,\mathbf{I}_d)}\|\nabla S(\mathbf{x})\|^2$. Thus, using the elementary inequality $\|a-b\|^2 \leq 2\|a\|^2 +2\|b\|^2$, we get that
$$
\E_{\mathbf{x}\sim\cN(0,\mathbf{I}_d)}\|\nabla Q(\mathbf{x})\|^2
\le
2\E_{\mathbf{x}\sim\cN(0,\mathbf{I}_d)}\|\nabla S(\mathbf{x})\|^2
+
2\E_{\mathbf{x}\sim\cN(0,\mathbf{I}_d)}\|\nabla S_V(\mathbf{x})\|^2
\le
O(\tau^{-1}).
$$
Moreover, $S$ depends only on the direction $\mathbf{w}$, while $S_V$ depends only on the projected direction $\mathbf{w}_V$. Hence $Q=S-S_V$ depends only on $\operatorname{span}\{\mathbf{w},\mathbf{w}_V\}$, and therefore $\rank(\mathbf{M}(Q))\le 2$. By \Cref{lem:nuclear-rank-trace},
$$
\tr(\mathbf{M}(Q)^{1/2})
\le
\sqrt{\rank(\mathbf{M}(Q))\tr(\mathbf{M}(Q))}
=
\sqrt{\rank(\mathbf{M}(Q))\,\E_{\mathbf{x}\sim\cN(0,\mathbf{I}_d)}\|\nabla Q(\mathbf{x})\|^2}
\le
O(\tau^{-1/2}).
$$
This proves Eq.\;\eqref{eq:Q-nuclear-bound}.
\end{proof}

For a polynomial $Q\in\cP_k$, write
\begin{equation}\label{eq:gain-definition}
\mathsf{Gain}_P(Q)
:=
\E_{(\mathbf{x},y)\sim D}[Q(\mathbf{x})\,y\,\psi(yP(\mathbf{x}))]
-
2\mu\E_{\mathbf{x}\sim\cN(0,\mathbf{I}_d)}[P(\mathbf{x})Q(\mathbf{x})]
-
\nu\,\tr\!\big(\mathbf{M}(Q)^{1/2}\big).
\end{equation}

Intuitively, $\mathsf{Gain}_P(Q)$ is the first-order change of how fast the regularized loss $\cL_{\mu,\nu}$ changes in the direction $Q$: approximate optimality of $P$ forces it to be small, while our constructed $Q$ would have large positive gain if averaging $f$ to $f_V$ lost too much label correlation.

We next combine the previous estimates to lower bound the gain of the polynomial $Q$ constructed in \Cref{lem:q-approx}. By \Cref{lem:q-approx}, \Cref{lem:logistic-gradient-correlation}, and the bound $|y\psi(yP(\mathbf{x}))|\le 1$,
\begin{equation}\label{eq:Q-logistic-correl}
\E_{(\mathbf{x},y)\sim D}[Q(\mathbf{x})\,y\,\psi(yP(\mathbf{x}))]
\ge
\frac{1}{2}\E_{(\mathbf{x},y)\sim D}[(f(\mathbf{x})-f_V(\mathbf{x}))y]
-
\frac{1}{2}\sqrt{\eta}
-
2\tau.
\end{equation}
Combining Eq.\;\eqref{eq:Q-logistic-correl}, \eqref{eq:PQ-bound}, and \eqref{eq:Q-nuclear-bound} with the definition of $\mathsf{Gain}_P(Q)$ in \eqref{eq:gain-definition} gives
\begin{equation}\label{eq:gain-lower}
\mathsf{Gain}_P(Q)
\ge
\frac{1}{2}\E_{(\mathbf{x},y)\sim D}[(f(\mathbf{x})-f_V(\mathbf{x}))y]
-
\left(\frac{1}{2}+2\mu\sqrt5\right)\sqrt{\eta}
-
2\tau
-
O(\nu\tau^{-1/2}).
\end{equation}

\begin{lemma}[Approximate optimality implies no large gain]\label{lem:no-large-gain}
Let $P\in\cP_k$ satisfy
$
\cL_{\mu,\nu}(P)\le \inf_{P'\in\cP_k}\cL_{\mu,\nu}(P')+\alpha.
$
If $Q\in\cP_k$ and $\E_{\mathbf{x}\sim\cN(0,\mathbf{I}_d)}[Q(\mathbf{x})^2]\le 5$, then
$$
\mathsf{Gain}_P(Q)\le O(\sqrt{\alpha}).
$$
\end{lemma}

\begin{proof}
By the standard second-order Taylor bound, since $\ell''\le 1/4$, for any $a,b$,
\begin{equation}\label{eq:taylor-smoothness}
\ell(b)\le \ell(a)+\ell'(a)(b-a)+\frac{1}{8}(b-a)^2.
\end{equation}
Apply Eq.\! \eqref{eq:taylor-smoothness} pointwise with $a=yP(\mathbf{x})$ and $b=y(P(\mathbf{x})+\zeta Q(\mathbf{x}))$ and take expectation. Since $\ell'(u)=-\psi(u)$, we get
$$
\E_{(\mathbf{x},y)\sim D}[\ell(y(P(\mathbf{x})+\zeta Q(\mathbf{x})))]
\le
\E_{(\mathbf{x},y)\sim D}[\ell(yP(\mathbf{x}))]
-\zeta\E_{(\mathbf{x},y)\sim D}[Q(\mathbf{x})y\psi(yP(\mathbf{x}))]
+\frac{1}{8}\zeta^2\E_{\mathbf{x}\sim\cN(0,\mathbf{I}_d)}[Q(\mathbf{x})^2].
$$
Adding the regularization terms and using
$
\tr(M(P+\zeta Q)^{1/2})
\le
\tr(\mathbf{M}(P)^{1/2})
+
\zeta\tr(\mathbf{M}(Q)^{1/2}),
$
we get
$$
\cL_{\mu,\nu}(P+\zeta Q)
\le
\cL_{\mu,\nu}(P)
-\zeta \mathsf{Gain}_P(Q)
+
\left(\mu+\frac{1}{8}\right)\zeta^2\E_{\mathbf{x}\sim\cN(0,\mathbf{I}_d)}[Q(\mathbf{x})^2].
$$
Choosing
$
\zeta=\frac{\mathsf{Gain}_P(Q)}{2\left(\mu+\frac{1}{8}\right)\E_{\mathbf{x}\sim\cN(0,\mathbf{I}_d)}[Q(\mathbf{x})^2]}
$
gives
$$
\cL_{\mu,\nu}(P+\zeta Q)
\le
\cL_{\mu,\nu}(P)
-
\frac{\mathsf{Gain}_P(Q)^2}
{4\left(\mu+\frac{1}{8}\right)\E_{\mathbf{x}\sim\cN(0,\mathbf{I}_d)}[Q(\mathbf{x})^2]}.
$$
Since $P$ is $\alpha$-approximately optimal, the decrease is at most $\alpha$. Using $\E_{\mathbf{x}\sim\cN(0,\mathbf{I}_d)}[Q(\mathbf{x})^2]\le 5$, we get
$$
\mathsf{Gain}_P(Q)\le 2\sqrt{5\left(\mu+\frac{1}{8}\right)\alpha}.
$$
\end{proof}

We now finish the proof of halfspace-averaging bound \eqref{eq:averaging-main-bound}. Choose $\tau:=\eps/32$. By the choice of the sufficiently large constant $C$ in $k=C/\eps^2$, the condition $k\ge C_{\mathrm{HS}}\tau^{-2}$ holds. Apply \Cref{lem:no-large-gain} with $\alpha=O(\eps^3)$. By taking the hidden constant in the optimization accuracy sufficiently small, as allowed in the theorem statement, we get
$$
\mathsf{Gain}_P(Q)\le \frac{\eps}{32}.
$$
On the other hand, Eq.\! \eqref{eq:gain-lower} gives
$$
\mathsf{Gain}_P(Q)
\ge
\frac{1}{2}\E_{(\mathbf{x},y)\sim D}[(f(\mathbf{x})-f_V(\mathbf{x}))y]
-
\left(\frac{1}{2}+2\mu\sqrt5\right)\sqrt{\eta}
-
2\tau
-
O(\nu\tau^{-1/2}).
$$
Since $\eta=\eps^2/1024$, $\mu=1/128$, and $\nu=c_\nu\eps^{3/2}$ with $c_\nu>0$ sufficiently small, we get
$$
\frac{1}{2}\E_{(\mathbf{x},y)\sim D}[(f(\mathbf{x})-f_V(\mathbf{x}))y]
\le
\frac{\eps}{32}
+
\left(\frac{1}{2}+2\mu\sqrt5\right)\sqrt{\eta}
+
2\tau
+
O(\nu\tau^{-1/2})
\le
\frac{\eps}{2}.
$$
Therefore, we obtain the halfspace-averaging bound \eqref{eq:averaging-main-bound}:
$$
\E_{(\mathbf{x},y)\sim D}[(f(\mathbf{x})-f_V(\mathbf{x}))y]\le \eps,
$$

It remains to extract a single proper halfspace in $V$. If $\mathbf{w}_{V^\perp}=0$, then $f\in\mathcal{C}_V$ and the claim is trivial. Otherwise, by \Cref{app:def:halfspace-averaging}, $f_V(\mathbf{x})=\E_{z\sim\cN(0,1)}[h_z(\mathbf{x})]$ with $h_z\in\mathcal{C}_V$ for every $z$. Hence
$
\E_{z\sim\cN(0,1)}\E_{(\mathbf{x},y)\sim D}[(f(\mathbf{x})-h_z(\mathbf{x}))y]
=
\E_{(\mathbf{x},y)\sim D}[(f(\mathbf{x})-f_V(\mathbf{x}))y]
\le
\eps.
$
Therefore there exists some $z$ such that $\E_{(\mathbf{x},y)\sim D}[(f(\mathbf{x})-h_z(\mathbf{x}))y]\le \eps$. Since $h_z\in\mathcal{C}_V$, taking $h=h_z$ proves the second claim.
\end{proof}

\subsection{Finite Search in the Reduced Subspace for Halfspaces}

We now show that the low-dimensional subspace $V$ from \Cref{app:thm:structural} has dimension independent of $d$. This allows us to replace the search over $\R^d$ by a finite search over a net in $V$.

\begin{lemma}[Dimension bound]\label{lem:dim-bound}
Let $V$ be the span of the eigenvectors of $\mathbf{M}=\mathbf{M}(P)$ with eigenvalues at least $\eta>0$, and suppose that $P$ is an $O(\eps^3)$-approximate minimizer of $\cL_{\mu,\nu}$. Then
$$
\dim(V)\le \frac{\tr(\mathbf{M}^{1/2})}{\sqrt{\eta}}
=
O\!\left(\frac{1}{\nu\sqrt{\eta}}\right).
$$
\end{lemma}

\begin{proof}
Let $\lambda_1,\dots,\lambda_r$ be the eigenvalues of $\mathbf{M}$ corresponding to $V$, where $r=\dim(V)$. Since $\lambda_i\ge \eta$,
$$
\dim(V)\sqrt{\eta}
=
r\sqrt{\eta}
\le
\sum_{i=1}^r\sqrt{\lambda_i}
\le
\tr(\mathbf{M}^{1/2}).
$$
Also, since $0\in\cP_k$ and the regularization terms are nonnegative,
$$
\nu\tr(\mathbf{M}^{1/2})
\le
\cL_{\mu,\nu}(P)
\le
\cL_{\mu,\nu}(0)+O(\eps^3)
=
\log 2+O(\eps^3)
=
O(1).
$$
Thus $\tr(\mathbf{M}^{1/2})=O(1/\nu)$, and the claim follows.
\end{proof}

To turn the existence of a low-dimensional subspace into a proper learner, we use the same
discretization facts as in \cite[Section~3.1]{DKTKZ21}; see also
\cite[Corollary~4.2.13]{Ver18} for the covering bound. In particular, we construct a standard $\eps$-net
$\widetilde V \subseteq \{\mathbf{v} \in V:\ \|\mathbf{v}\|_2 = 1 \}$ for the unit sphere of the low-dimensional subspace $V$, together with a discrete set
$\widetilde T$ of thresholds $t\in\R$. This procedure results in the  $\eps$-cover of candidate halfspaces
$$
\mathcal H:
=\{\sign(\langle \mathbf{v},\mathbf{x}\rangle+t):\, \mathbf{v}\in\widetilde V,\ t\in\widetilde T\}.
$$

\begin{fact}[Covering and threshold discretization]\label{fact:cover}
Let $V\subseteq \R^d$ be a subspace with $r = \dim(V)$. For every $\eps\in(0,1/2)$ and every halfspace $h(\mathbf{x})=\sign(\ip{\mathbf{v}}{\mathbf{x}}+t)$ with $\mathbf{v}\in V$ and $\|\mathbf{v}\|_2=1$, there exists a halfspace $\widetilde h$ in a  $\eps$-cover $\cH$ of size $(1/\eps)^{O(r)}$ satisfying
$$
\Pr_{\mathbf{x}\sim\cN(0,\mathbf{I}_d)}[h(\mathbf{x})\neq \widetilde h(\mathbf{x})] \le O(\eps).
$$
\end{fact}

We can now prove \Cref{thm:main}.

\begin{theorem}[Proper learning of halfspaces]\label{thm:main}
There exists a proper agnostic learner for halfspaces under Gaussian marginals with the following guarantee. For every $\eps,\delta \in (0,1/2)$ and every distribution $D$ over $(\mathbf{x},y)\in\R^d\times\{\pm1\}$ whose $\mathbf{x}$-marginal is $\cN(0,\mathbf{I}_d)$, the learner draws
$
N = d^{O(1/\eps^2)}\,\log(1/\delta)
$
examples, runs in time
$
d^{O(1/\eps^2)}\,\log(1/\delta)
+ (1/\eps)^{O(1/\eps^{2.5})}\,\log(1/\delta),
$
and returns a halfspace $h$ such that, with probability at least $1-\delta$,
$
\err_D(h) \le \OPT + \eps.
$
\end{theorem}

\begin{proof}[Proof of \Cref{thm:main}]
Set $k=\Theta(\eps^{-2})$, $\eta=\Theta(\eps^2)$, $\mu:=1/128$, and $\nu=\Theta(\eps^{3/2})$. By \Cref{app:prop:compute-P}, using
$$
N_1 = d^{O(k)}\,\poly(1/\eps)\,\log(1/\delta)
= d^{O(1/\eps^2)}\,\poly(1/\eps)\,\log(1/\delta)
$$
samples we compute a degree-$k$ polynomial $P$ such that $ \cL_{\mu,\nu}(P) \le \inf_{P'\in\cP_k}\cL_{\mu,\nu}(P') + O(\eps^3)$ with probability at least $1-\delta/3$.

Form the influence matrix $\mathbf{M}(P) = \E[\nabla P(\mathbf{x}) \nabla P(\mathbf{x})^\top]$, let $V$ be the span of the eigenvectors with eigenvalues at least $\eta$, and let $r=\dim(V)$. By \Cref{lem:dim-bound},
$$
r = O\left(\frac{1}{\nu \sqrt{\eta}}\right) = O(\eps^{-2.5}).
$$
Fix any halfspace $f\in\mathcal{C}$ with $\err_D(f)\le \OPT + \eps$. By \Cref{app:thm:structural}, there exists $h_V\in\mathcal{C}_V$ such that
$$
\E[(f-h_V)y] \le \eps,
$$
hence
$$
\err_D(h_V) \le \err_D(f) + \eps/2 \le \OPT + 3\eps/2.
$$
Applying \Cref{fact:cover}, we obtain a  $\eps$-cover $\cH$ of halfspaces of size
$
|\cH| = (1/\eps)^{O(r)} = (1/\eps)^{O(1/\eps^{2.5})}
$
containing a hypothesis $\widetilde h$ with
$$
\err_D(\widetilde h) \le \err_D(h_V)+O(\eps) \le \OPT + O(\eps).
$$

Finally, let
$
h=\arg\min_{g\in\cH}\widehat{\err}(g),
$
where $\widehat{\err}$ is computed on an independent validation sample of size
$
N_2 = O\!\left(\frac{\log|\cH|+\log(1/\delta)}{\eps^2}\right).
$
By Hoeffding's inequality and a union bound over the $\eps$-cover $\cH$, with probability at least $1-\delta/3$,
$$
\err_D(h)\le \min_{g\in\cH}\err_D(g)+\eps\le \OPT+O(\eps).
$$
Rescaling $\eps$ by a sufficiently small absolute constant gives the stated $\OPT+\eps$ guarantee.

Combining this with the success event of \Cref{app:prop:compute-P} completes the proof.

The running time is the sum of the time needed to optimize over the degree-$k$ Hermite coefficients, compute the eigenspace of $\mathbf{M}(P)$, and search over the  $\eps$-cover $\cH$. The first term is $d^{O(1/\eps^2)}\log(1/\delta)$ by \Cref{app:prop:compute-P}, and the second is $(1/\eps)^{O(1/\eps^{2.5})}\log(1/\delta)$.
\end{proof}

\section{Proper Learning of Boolean Functions of Halfspaces}\label{subsec:boolean-halfspaces}

In this section, we analyze \Cref{alg:proper-boolean} and show that with high probability it gives a Boolean function of $K$ halfspaces whose prediction accuracy is $\eps$-close to the optimal accuracy $\OPT_K$ among all such functions. 

We denote by $\mathcal{C}_K$ the class of arbitrary Boolean functions of $K$ halfspaces
$$
\mathcal{C}_{K}
:=
\Big\{
\mathbf{x} \mapsto B(f_1(\mathbf{x}),\dots,f_K(\mathbf{x}))
:
f_1,\dots,f_K \in \mathcal{C}, \, B: \{-1, +1\}^K \to \{-1,+1\}
\Big\}.
$$
If $V\subseteq \R^d$ is a subspace, we write $\mathcal{C}_{K,V}$ for the corresponding subclass in which all halfspace normal vectors lie in $V$.

We first compute a near-minimizer of the regularized logistic-loss objective over low-degree polynomials.

\begin{proposition}[Regularized logistic-loss polynomial regression]\label{prop:compute-P}
Fix $\mu>0$ to be an absolute constant and let $\nu\ge 0$. Let $D$ be a distribution over $(\mathbf{x},y)\in\R^d\times\{\pm1\}$ whose $\mathbf{x}$-marginal is $\cN(0,\mathbf{I}_d)$. Let $k\in\mathbb Z_+$ and $\eps,\delta>0$. There is an algorithm that draws
$
N=(dk)^{O(k)}\frac{\log(1/\delta)}{\eps^2}
$
samples from $D$, runs in time $\poly(N,d)$, and outputs a polynomial $P\in\cP_k$ such that
$
\cL_{\mu,\nu}(P)\le \inf_{P'\in\cP_k}\cL_{\mu,\nu}(P')+\eps
$
with probability at least $1-\delta$.
\end{proposition}

A proof can be found in Appendix~\ref{app:compute_reg}.

Klivans et al. \cite{KOS08} showed that Gaussian surface area is a useful instrument for learning broad families of geometric concepts. 
We recall the definition.

\begin{definition}[Gaussian surface area]\label{def:gsa}
Let $\gamma_d$ denote the standard Gaussian measure on $\R^d$. For a measurable set $A\subseteq\R^d$ and $\delta>0$, define
$A_\delta:=\{\mathbf{x}\in\R^d:\operatorname{dist}(\mathbf{x},A)\le \delta\}$.
The Gaussian surface area of $A$ is
$$
\GSA(A):=\liminf_{\delta\to 0^+}\frac{\gamma_d(A_\delta\setminus A)}{\delta}.
$$
\end{definition}

For Boolean functions of halfspaces, we use the following Gaussian surface area bound.

\begin{fact}[GSA bound for Boolean functions of $K$ halfspaces {\normalfont\protect\cite[Fact~17]{KOS08}}]\label{fact:gsa-boolean}
There exists an absolute constant $C>0$ such that for every $K\ge 1$ and every $f\in\mathcal{C}_{K}$,
$$
\GSA(f)\le O(K).
$$
\end{fact}

One can view Gaussian surface area as a quantitative measure of the regularity of the decision boundary, which we use to obtain the required low-degree approximation. For a single halfspace, the one-dimensional FT-mollification argument of \cite{DKN10} gives a sharper approximation bound, as discussed in \Cref{halfspace}. For multidimensional concepts, such as Boolean functions of halfspaces, the decision boundary is no longer determined by a single normal direction. We therefore use a Hermite-concentration argument based on the Ornstein--Uhlenbeck operator and Gaussian surface area, which applies to general decision boundaries.

\begin{theorem}[Gaussian $L_1$ approximation for Boolean functions of halfspaces]\label{thm:boolean-approx}
There exist absolute constants $C_{\mathrm{BH}},C'_{\mathrm{BH}}>0$ with the following property. For every integer $K\ge 1$, every $\tau\in(0,1/10)$, and every $f\in\mathcal{C}_{K}$, there exists a polynomial $S:\R^d\to\R$ of degree at most $C_{\mathrm{BH}}\frac{K^2\log(1/\tau)}{\tau^2}$ such that
$$
\E_{\mathbf{x}\sim\cN(0,\mathbf{I}_d)}[\,|f(\mathbf{x})-S(\mathbf{x})|\,]\le \tau,
$$
$$
\E_{\mathbf{x}\sim\cN(0,\mathbf{I}_d)}\|\nabla S(\mathbf{x})\|^2\le C'_{\mathrm{BH}}\frac{K^2}{\tau},
$$
$$
\E_{\mathbf{x}\sim\cN(0,\mathbf{I}_d)}[S(\mathbf{x})^2]\le 1.
$$
\end{theorem}

A proof can be found in Appendix~\ref{app:multi-ou}. Compared with the one-dimensional halfspace approximation \Cref{app:thm:threshold-approx}, this argument incurs an additional logarithmic factor $\log(1/\tau)$ in the polynomial degree.

We also use the following lemma.
\begin{lemma}[Subspace Gaussian Poincar\'e]\label{lem:subspace-poincare}
Let $W\subseteq \R^d$ be a subspace of dimension $r$, and let $\{\boldsymbol{\xi}_1,\dots,\boldsymbol{\xi}_r\}$ be an orthonormal basis of $W$. Write $\mathbf{x}=\mathbf{u}+\mathbf{z}$ with $\mathbf{u}\in W^\perp$ and $\mathbf{z}=\sum_{i=1}^r z_i\boldsymbol{\xi}_i$, where $\mathbf{u}$ and $z_1,\dots,z_r$ are independent and $z_1,\dots,z_r \stackrel{\mathrm{i.i.d.}}{\sim} \cN(0,1)$. Then for every weakly differentiable $G:\R^d\to\R$,
$$
\E\!\left[\Big(G(\mathbf{x})-\E[G(\mathbf{x})\mid \mathbf{u}]\Big)^2\right]
\le
\sum_{i=1}^r \E\big[(\partial_{\boldsymbol{\xi}_i}G(\mathbf{x}))^2\big].
$$
\end{lemma}

A proof can be found in Appendix~\ref{app:gp}. We will use this concentration inequality to show that averaging over a low-influence direction changes the relevant quantities by only a small amount.

\subsection{Reducing to a Low-Dimensional Search Space}\label{sec:existence-guarantee}

%We now extend the low-dimensional subspace argument to arbitrary Boolean functions of halfspaces. 
The first step in our analysis is to define the appropriate averaging operator.
Let $V\subseteq\R^d$ be a subspace and let
$
W:=\operatorname{span}\{(\mathbf{w}_1)_{V^\perp},\dots,(\mathbf{w}_K)_{V^\perp}\}\subseteq V^\perp
$. Let  $\{\boldsymbol{\xi}_1,\dots,\boldsymbol{\xi}_r\}$ be an orthonormal basis of $W$, where $r=\dim(W)\le K$. Decompose $\mathbf{x}=\mathbf{x}_{W^\perp}+\mathbf{x}_{W}$, where $\mathbf{x}_{W^\perp}$ and $\mathbf{x}_{W}$ are the projections of $\mathbf{x}$ onto $W^\perp$ and $W$, respectively. For
$
\mathbf{z}=\sum_{i=1}^r z_i \boldsymbol{\xi}_i
$
with $z_i\stackrel{\mathrm{i.i.d.}}{\sim}\cN(0,1)$ independent of $\mathbf{x}_{W^\perp}$, we have a resampling identity 
\begin{equation}\label{eq:B-resample}
\mathbf{x}=\mathbf{x}_{W^\perp}+\mathbf{x}_{W} \overset{\mathrm d}{=} \mathbf{x}_{W^\perp}+\mathbf{z}.
\end{equation}

This is analogous to the one-dimensional resampling identity used for a single halfspace in \eqref{eq:resample}, except that now $\mathbf{z}$ is a Gaussian vector lying in the subspace $W$.

\begin{definition}[Averaging a Boolean function of halfspaces]\label{def:boolean-averaging}

For each $\mathbf{z}\in W$, define the shifted Boolean function of halfspaces $h_{\mathbf{z}}(\mathbf{x}):=f(\mathbf{x}_{W^\perp}+\mathbf{z})$. Based on the resampling identity \eqref{eq:B-resample}, define
\begin{equation}\label{eq:B-averaged}
f_{V}(\mathbf{x}):=\E_{\mathbf{z}}[h_{\mathbf{z}}(\mathbf{x})] = \E_{\mathbf{z}}[f(\mathbf{x}_{W^\perp}+\mathbf{z})].
\end{equation}
If $W=\{0\}$, define $f_{V}(\mathbf{x}):=f(\mathbf{x})$.
\end{definition}

For each fixed $\mathbf{z}\in W$, the shifted function $h_{\mathbf{z}}$ belongs to $\mathcal{C}_{K,V}$. Indeed, for each $j$, set $t_{j,\mathbf{z}}:=t_j+\langle (\mathbf{w}_j)_{V^\perp},\mathbf{z}\rangle$. Since $(\mathbf{w}_j)_{V^\perp}\in W$ and $(\mathbf{w}_j)_{V}\perp W$, we have
$$
f_j(\mathbf{x}_{W^\perp}+\mathbf{z}) = \sign(\langle (\mathbf{w}_j)_{V},\mathbf{x}\rangle+t_{j,\mathbf{z}}).
$$
Thus each shifted halfspace has normal vector in $V$. Applying the same Boolean map $B$ gives $h_{\mathbf{z}}\in\mathcal{C}_{K,V}$ for every fixed $\mathbf{z}$, and therefore $f_{V}\in\conv(\mathcal{C}_{K,V})$.

\begin{remark}
If $(\mathbf{w}_j)_{V}=0$, the shifted halfspace is a constant function, which we regard as a degenerate halfspace. 
\end{remark}

\begin{theorem}[Low-dimensional reduction for Boolean functions of halfspaces]\label{thm:structural-B}
There exists an absolute constant $C>0$ and an absolute constant $c_\nu>0$ such that the following holds. Fix $K\ge 1$ and $\eps\in(0,1/10)$, and set $k:=CK^2\log(1/\eps)/\eps^2$, $\mu:=1/128$, $ \eta:=\eps^2/(1024K)$, and $\nu:=c_\nu\eps^{3/2}/K^{3/2}$.
Let $P\in\cP_k$ satisfy $\cL_{\mu,\nu}(P)\le \inf_{P'\in\cP_k}\cL_{\mu,\nu}(P')+O(\eps^3)$, where the hidden constant is sufficiently small. Let $V$ be the span of the eigenvectors of $\mathbf{M}(P)$ whose eigenvalues are at least $\eta$. Then for every $f\in\mathcal{C}_{K}$, the averaged function $f_V$ from \Cref{def:boolean-averaging} satisfies
\begin{equation}\label{eq:B-averaging-main-bound}
\E_{(\mathbf{x},y)\sim D}[(f(\mathbf{x})-f_V(\mathbf{x}))y]\le \eps.
\end{equation}
Consequently, there exists $h\in\mathcal{C}_{K,V}$ such that
$$
\E_{(\mathbf{x},y)\sim D}[(f(\mathbf{x})-h(\mathbf{x}))y]\le \eps.
$$
\end{theorem}

\begin{proof}
Fix $f\in\mathcal{C}_{K}$. If $W=\{0\}$, then $f_V=f$, and the bound \eqref{eq:B-averaging-main-bound} is immediate. Hence assume $W\neq\{0\}$.

The next lemma controls the correlation between the averaging residual $f-f_{V}$ and the logistic gradient, i.e. the gradient of the logistic loss at the margin $yP$.

\begin{lemma}[Correlation with the logistic gradient]\label{lem:B-logistic-gradient-correlation}
For the averaged function $f_{V}$ from \Cref{def:boolean-averaging},
\begin{equation}\label{eq:B-g-logistic-grad}
\E_{(\mathbf{x},y)\sim D}[(f(\mathbf{x})-f_{V}(\mathbf{x}))\,y\,\psi(yP(\mathbf{x}))]
\ge
\frac{1}{2}\E_{(\mathbf{x},y)\sim D}[(f(\mathbf{x})-f_{V}(\mathbf{x}))y]
-
\frac{1}{2}\sqrt{K\eta}.
\end{equation}
\end{lemma}

\begin{proof}
Using \Cref{lem:label-flip-decomposition}, \begin{equation}\label{eq:B-symmetric-corollary} \E_{(\mathbf{x},y)\sim D}[(f(\mathbf{x})-f_{V}(\mathbf{x}))\,y\,\psi(yP(\mathbf{x}))] = \frac{1}{2}\E_{(\mathbf{x},y)\sim D}[(f(\mathbf{x})-f_{V}(\mathbf{x}))y] - \E_{\mathbf{x}\sim\cN(0,\mathbf{I}_d)}[(f(\mathbf{x})-f_{V}(\mathbf{x}))\phi(P(\mathbf{x}))]. \end{equation}
We first bound the model-dependent term in Eq.\;\eqref{eq:B-symmetric-corollary}.  The estimate we need is
\begin{equation}\label{eq:B-gphi-bound}
\E_{\mathbf{x}\sim\cN(0,\mathbf{I}_d)}[(f(\mathbf{x})-f_{V}(\mathbf{x}))\phi(P(\mathbf{x}))]
\le
\frac{1}{2}\sqrt{K\eta}.
\end{equation}

Let $G(\mathbf{x}):=\phi(P(\mathbf{x}))$ and
$
G_{V}(\mathbf{x}):=\E_{\mathbf{z}}[G(\mathbf{x}_{W^\perp}+\mathbf{z})].
$
Since $G_{V}(\mathbf{x})$ depends only on $\mathbf{x}_{W^\perp}$,
$$
\E_{\mathbf{x}\sim\cN(0,\mathbf{I}_d)}[(f(\mathbf{x})-f_{V}(\mathbf{x}))G_{V}(\mathbf{x})]
=
\E_{\mathbf{x}\sim\cN(0,\mathbf{I}_d)}
\!\left[
G_{V}(\mathbf{x})\,
\E_{\mathbf{x}\sim\cN(0,\mathbf{I}_d)}
[f(\mathbf{x})-f_{V}(\mathbf{x})\mid \mathbf{x}_{W^\perp}]
\right]
=
0,
$$
where the last equality follows from the definition of $f_{V}$. Hence
$$
\E_{\mathbf{x}\sim\cN(0,\mathbf{I}_d)}[(f(\mathbf{x})-f_{V}(\mathbf{x}))\phi(P(\mathbf{x}))]
=
\E_{\mathbf{x}\sim\cN(0,\mathbf{I}_d)}[(f(\mathbf{x})-f_{V}(\mathbf{x}))(G(\mathbf{x})-G_{V}(\mathbf{x}))].
$$
By \Cref{lem:subspace-poincare},
$$
\E_{\mathbf{x}\sim\cN(0,\mathbf{I}_d)}[(G(\mathbf{x})-G_{V}(\mathbf{x}))^2]
\le
\sum_{i=1}^r
\E_{\mathbf{x}\sim\cN(0,\mathbf{I}_d)}[(\partial_{\boldsymbol{\xi}_i}G(\mathbf{x}))^2].
$$
Since $|\phi'|\le 1/4$,
$$
\sum_{i=1}^r
\E_{\mathbf{x}\sim\cN(0,\mathbf{I}_d)}[(\partial_{\boldsymbol{\xi}_i}G(\mathbf{x}))^2]
\le
\frac{1}{16}
\sum_{i=1}^r
\E_{\mathbf{x}\sim\cN(0,\mathbf{I}_d)}[(\partial_{\boldsymbol{\xi}_i}P(\mathbf{x}))^2]
=
\frac{1}{16}\sum_{i=1}^r \boldsymbol{\xi}_i^\top \mathbf{M}(P)\boldsymbol{\xi}_i
\le
\frac{r\eta}{16}
\le
\frac{K\eta}{16},
$$
because every $\boldsymbol{\xi}_i\in W\subseteq V^\perp$, and all eigenvalues of $\mathbf{M}(P)$ on $V^\perp$ are smaller than $\eta$. Finally, $|f(\mathbf{x})-f_{V}(\mathbf{x})|\le 2$, so the Cauchy--Schwarz inequality gives
$$
\E_{\mathbf{x}\sim\cN(0,\mathbf{I}_d)}
[(f(\mathbf{x})-f_{V}(\mathbf{x}))(G(\mathbf{x})-G_{V}(\mathbf{x}))]
\le
\|f-f_{V}\|_2\|G-G_{V}\|_2
\le
2\cdot\frac{\sqrt{K\eta}}{4}
=
\frac{1}{2}\sqrt{K\eta}.
$$
This proves Eq.\;\eqref{eq:B-gphi-bound}. Combining Eq.\;\eqref{eq:B-symmetric-corollary} with \eqref{eq:B-gphi-bound} gives Eq.\;\eqref{eq:B-g-logistic-grad}.
\end{proof}

\begin{lemma}[Low-degree approximation of $f-f_{V}$]\label{lem:B-q-approx}
Fix $\tau\in(0,1/10)$. Suppose
$
k\ge C_{\mathrm{BH}}\frac{K^2\log(1/\tau)}{\tau^2}.
$
Then there exists $Q\in\cP_k$ such that
\begin{equation}\label{eq:B-q-approx-bound}
\E_{\mathbf{x}\sim\cN(0,\mathbf{I}_d)}
\big[|Q(\mathbf{x})-(f(\mathbf{x})-f_{V}(\mathbf{x}))|\big]
\le 2\tau,
\end{equation}
\begin{equation}\label{eq:B-q-bound}
\E_{\mathbf{x}\sim\cN(0,\mathbf{I}_d)}[Q(\mathbf{x})^2]\le 1.
\end{equation}
\end{lemma}

\begin{proof}
By \Cref{thm:boolean-approx}, there exists a polynomial $S$ of degree at most
$
C_{\mathrm{BH}}\frac{K^2\log(1/\tau)}{\tau^2}
$
such that
$
\E_{\mathbf{x}\sim\cN(0,\mathbf{I}_d)}[|S(\mathbf{x})-f(\mathbf{x})|]\le \tau
$
and
$
\E_{\mathbf{x}\sim\cN(0,\mathbf{I}_d)}[S(\mathbf{x})^2]\le 1.
$
Since $k\ge C_{\mathrm{BH}}\frac{K^2\log(1/\tau)}{\tau^2}$, we have $S\in\cP_k$. Define
$
S_{V}(\mathbf{x}):=\E_{\mathbf{z}}[S(\mathbf{x}_{W^\perp}+\mathbf{z})]
$
and set
$
Q(\mathbf{x}):=S(\mathbf{x})-S_{V}(\mathbf{x}).
$
Gaussian averaging in the $W$-directions preserves degree, so $Q\in\cP_k$. Moreover,
$$
\E_{\mathbf{x}\sim\cN(0,\mathbf{I}_d)}[Q(\mathbf{x})^2]
=
\E_{\mathbf{x}\sim\cN(0,\mathbf{I}_d)}
[\Var(S(\mathbf{x})\mid \mathbf{x}_{W^\perp})]
\le
\E_{\mathbf{x}\sim\cN(0,\mathbf{I}_d)}[S(\mathbf{x})^2]
\le
1.
$$
Also,
$
Q(\mathbf{x})-(f(\mathbf{x})-f_{V}(\mathbf{x}))
=
(S(\mathbf{x})-f(\mathbf{x}))-(S_{V}(\mathbf{x})-f_{V}(\mathbf{x})).
$
Hence
$$
\E_{\mathbf{x}\sim\cN(0,\mathbf{I}_d)}
\big[|Q(\mathbf{x})-(f(\mathbf{x})-f_{V}(\mathbf{x}))|\big]
\le
\E_{\mathbf{x}\sim\cN(0,\mathbf{I}_d)}[|S(\mathbf{x})-f(\mathbf{x})|]
+
\E_{\mathbf{x}\sim\cN(0,\mathbf{I}_d)}[|S_{V}(\mathbf{x})-f_{V}(\mathbf{x})|].
$$
By Jensen's inequality,
$$
\E_{\mathbf{x}\sim\cN(0,\mathbf{I}_d)}[|S_{V}(\mathbf{x})-f_{V}(\mathbf{x})|]
\le
\E_{\mathbf{x}\sim\cN(0,\mathbf{I}_d)}
\!\left[
\E_{\mathbf{z}}[|S(\mathbf{x}_{W^\perp}+\mathbf{z})-f(\mathbf{x}_{W^\perp}+\mathbf{z})|]
\right]
=
\E_{\mathbf{x}\sim\cN(0,\mathbf{I}_d)}[|S(\mathbf{x})-f(\mathbf{x})|],
$$
where the last equality uses the Gaussian resampling identity
$
\mathbf{x}\overset{d}=\mathbf{x}_{W^\perp}+\mathbf{z}.
$
Therefore, the bound \eqref{eq:B-q-approx-bound} follows.
\end{proof}

\begin{lemma}[Control of the regularization terms]\label{lem:B-regularizer-term}
Let $Q=S-S_{V}$ be the polynomial constructed in \Cref{lem:B-q-approx}. Then
\begin{equation}\label{eq:B-PQ-bound}
\E_{\mathbf{x}\sim\cN(0,\mathbf{I}_d)}[P(\mathbf{x})Q(\mathbf{x})]
\le
\sqrt{K\eta}.
\end{equation}
\begin{equation}\label{eq:B-Q-nuclear-bound}
\tr(\mathbf{M}(Q)^{1/2})
\le
O\!\left(\frac{K^{3/2}}{\sqrt{\tau}}\right).
\end{equation}
\end{lemma}

\begin{proof}
Recall that $Q(\mathbf{x})=S(\mathbf{x})-S_{V}(\mathbf{x})$, where
$
S_{V}(\mathbf{x})=\E_{\mathbf{z}}[S(\mathbf{x}_{W^\perp}+\mathbf{z})].
$
Hence
$
\E_{\mathbf{x}\sim\cN(0,\mathbf{I}_d)}[Q(\mathbf{x})\mid \mathbf{x}_{W^\perp}]=0.
$
Let
$
P_{V}(\mathbf{x}):=\E_{\mathbf{z}}[P(\mathbf{x}_{W^\perp}+\mathbf{z})].
$
Since $P_{V}(\mathbf{x})$ depends only on $\mathbf{x}_{W^\perp}$, we have
$
\E_{\mathbf{x}\sim\cN(0,\mathbf{I}_d)}[P_{V}(\mathbf{x})Q(\mathbf{x})]=0,
$
and therefore
$$
\E_{\mathbf{x}\sim\cN(0,\mathbf{I}_d)}[P(\mathbf{x})Q(\mathbf{x})]
=
\E_{\mathbf{x}\sim\cN(0,\mathbf{I}_d)}
[(P(\mathbf{x})-P_{V}(\mathbf{x}))Q(\mathbf{x})].
$$
By \Cref{lem:subspace-poincare},
$$
\E_{\mathbf{x}\sim\cN(0,\mathbf{I}_d)}[(P(\mathbf{x})-P_{V}(\mathbf{x}))^2]
\le
\sum_{i=1}^r
\E_{\mathbf{x}\sim\cN(0,\mathbf{I}_d)}[(\partial_{\boldsymbol{\xi}_i}P(\mathbf{x}))^2]
=
\sum_{i=1}^r \boldsymbol{\xi}_i^\top \mathbf{M}(P)\boldsymbol{\xi}_i
\le
r\eta
\le
K\eta.
$$
Using Eq.\;\eqref{eq:B-q-bound} and the Cauchy--Schwarz inequality gives
$$
\E_{\mathbf{x}\sim\cN(0,\mathbf{I}_d)}[P(\mathbf{x})Q(\mathbf{x})]
\le
\sqrt{
\E_{\mathbf{x}\sim\cN(0,\mathbf{I}_d)}[(P(\mathbf{x})-P_{V}(\mathbf{x}))^2]\,
\E_{\mathbf{x}\sim\cN(0,\mathbf{I}_d)}[Q(\mathbf{x})^2]
}
\le
\sqrt{K\eta}.
$$
This proves the first bound\;\eqref{eq:B-PQ-bound}.

We now prove the second bound\;\eqref{eq:B-Q-nuclear-bound}. By \Cref{thm:boolean-approx}, $S$ also satisfies
$
\E_{\mathbf{x}\sim\cN(0,\mathbf{I}_d)}\|\nabla S(\mathbf{x})\|^2
\le
O(K^2/\tau).
$
Since $S_{V}(\mathbf{x})=\E_{\mathbf{z}}[S(\mathbf{x}_{W^\perp}+\mathbf{z})]$, Jensen's inequality gives
$$
\E_{\mathbf{x}\sim\cN(0,\mathbf{I}_d)}\|\nabla S_{V}(\mathbf{x})\|^2
\le
\E_{\mathbf{x}\sim\cN(0,\mathbf{I}_d)}\|\nabla S(\mathbf{x})\|^2.
$$
Thus, using the elementary inequality $\|\mathbf{a}-\mathbf{b}\|^2\le 2\|\mathbf{a}\|^2+2\|\mathbf{b}\|^2$,
$$
\E_{\mathbf{x}\sim\cN(0,\mathbf{I}_d)}\|\nabla Q(\mathbf{x})\|^2
\le
2\E_{\mathbf{x}\sim\cN(0,\mathbf{I}_d)}\|\nabla S(\mathbf{x})\|^2
+
2\E_{\mathbf{x}\sim\cN(0,\mathbf{I}_d)}\|\nabla S_{V}(\mathbf{x})\|^2
\le
O(K^2/\tau).
$$
Moreover, by \Cref{thm:boolean-approx}, $S$ depends only on the span of the normals $\mathbf{w}_1,\dots,\mathbf{w}_K$, while $S_{V}$ depends only on the span of the projected normals $(\mathbf{w}_1)_{V},\dots,(\mathbf{w}_K)_{V}$. Hence $Q=S-S_{V}$ depends only on a subspace of dimension at most $2K$, and therefore $\rank(\mathbf{M}(Q))\le 2K$. By \Cref{lem:nuclear-rank-trace},
$$
\tr(\mathbf{M}(Q)^{1/2})
\le
\sqrt{\rank(\mathbf{M}(Q))\tr(\mathbf{M}(Q))}
=
\sqrt{\rank(\mathbf{M}(Q))\,\E_{\mathbf{x}\sim\cN(0,\mathbf{I}_d)}\|\nabla Q(\mathbf{x})\|^2}
\le
O\!\left(\frac{K^{3/2}}{\sqrt{\tau}}\right).
$$
This proves the nuclear-norm bound\;\eqref{eq:B-Q-nuclear-bound}, and completes the proof of the lemma.
\end{proof}

For a polynomial $Q\in\cP_k$, write
\begin{equation}\label{eq:B-gain-definition}
\mathsf{Gain}_P(Q)
:=
\E_{(\mathbf{x},y)\sim D}[Q(\mathbf{x})\,y\,\psi(yP(\mathbf{x}))]
-
2\mu\E_{\mathbf{x}\sim\cN(0,\mathbf{I}_d)}[P(\mathbf{x})Q(\mathbf{x})]
-
\nu\,\tr\!\big(\mathbf{M}(Q)^{1/2}\big).
\end{equation}

Intuitively, $\mathsf{Gain}_P(Q)$ is the first-order approximation of how fast the regularized loss $\cL_{\mu,\nu}$ changes in the direction $Q$: approximate optimality of $P$ forces it to be small, while our constructed $Q$ would have large positive gain if $f_{V}$ lost too much label correlation.

\begin{lemma}[Approximate optimality implies no large gain]\label{lem:B-no-large-gain}
Let $P\in\cP_k$ satisfy
$
\cL_{\mu,\nu}(P)\le \inf_{P'\in\cP_k}\cL_{\mu,\nu}(P')+\alpha.
$
If $Q\in\cP_k$ and $\E_{\mathbf{x}\sim\cN(0,\mathbf{I}_d)}[Q(\mathbf{x})^2]\le 5$, then
$$
\mathsf{Gain}_P(Q)\le O(\sqrt{\alpha}).
$$
\end{lemma}

\begin{proof}
The proof is identical to the proof of \Cref{lem:no-large-gain}.
\end{proof}

We next combine the previous estimates to lower bound the gain of the polynomial $Q$ constructed in \Cref{lem:B-q-approx}. By \Cref{lem:B-logistic-gradient-correlation}, \Cref{lem:B-q-approx}, and the bound $|y\psi(yP(\mathbf{x}))|\le 1$,
\begin{equation}\label{eq:B-Q-logistic-correl}
\E_{(\mathbf{x},y)\sim D}[Q(\mathbf{x})\,y\,\psi(yP(\mathbf{x}))]
\ge
\frac{1}{2}\E_{(\mathbf{x},y)\sim D}[(f(\mathbf{x})-f_{V}(\mathbf{x}))y]
-
\frac{1}{2}\sqrt{K\eta}
-
2\tau.
\end{equation}
Combining Eq.\;\eqref{eq:B-Q-logistic-correl}, \eqref{eq:B-PQ-bound}, and \eqref{eq:B-Q-nuclear-bound} with the definition of $\mathsf{Gain}_P(Q)$ in \eqref{eq:B-gain-definition} gives
\begin{equation}\label{eq:B-gain-lower}
\mathsf{Gain}_P(Q)
\ge
\frac{1}{2}\E_{(\mathbf{x},y)\sim D}[(f(\mathbf{x})-f_{V}(\mathbf{x}))y]
-
\left(\frac{1}{2}+2\mu\right)\sqrt{K\eta}
-
2\tau
-
O\!\left(\nu\frac{K^{3/2}}{\sqrt{\tau}}\right).
\end{equation}

The final step is to establish the averaging bound\;\eqref{eq:B-averaging-main-bound}. Choose $\tau:=\eps/32$. By the choice of the sufficiently large constant $C$ in
$
k=C K^2\log(1/\eps)/\eps^2,
$
the condition
$
k\ge C_{\mathrm{BH}}\frac{K^2\log(1/\tau)}{\tau^2}
$
holds. Apply \Cref{lem:B-no-large-gain} with $\alpha=O(\eps^3)$. Since $\E_{\mathbf{x}\sim\cN(0,\mathbf{I}_d)}[Q(\mathbf{x})^2]\le 1\le 5$, and by choosing the hidden constant in the optimization accuracy sufficiently small as allowed in the premise of the theorem, we get
$$
\mathsf{Gain}_P(Q)\le \frac{\eps}{32}.
$$
Combining $\mathsf{Gain}_P(Q)\le \eps/32$ with the bound\;\eqref{eq:B-gain-lower}, and using $\eta=\eps^2/(1024K)$, $\mu=1/128$, $\tau=\eps/32$, and $\nu=c_\nu\eps^{3/2}/K^{3/2}$ with $c_\nu>0$ sufficiently small, we get
$$
\frac{1}{2}\E_{(\mathbf{x},y)\sim D}[(f(\mathbf{x})-f_{V}(\mathbf{x}))y]
\le
\frac{\eps}{32}
+
\left(\frac{1}{2}+2\mu\right)\sqrt{K\eta}
+
2\tau
+
O\!\left(\nu\frac{K^{3/2}}{\sqrt{\tau}}\right)
\le
\frac{\eps}{2}.
$$
This proves the averaging bound\;\eqref{eq:B-averaging-main-bound}:
$$
\E_{(\mathbf{x},y)\sim D}[(f(\mathbf{x})-f_{V}(\mathbf{x}))y]\le \eps.
$$

It remains to extract a single proper Boolean function of halfspaces in $V$. If $W=\{0\}$, then $f\in\mathcal{C}_{K,V}$ and the claim is trivial. Otherwise, by \Cref{def:boolean-averaging},
$
f_{V}(\mathbf{x})=\E_{\mathbf{z}}[h_{\mathbf{z}}(\mathbf{x})]
$
with $h_{\mathbf{z}}\in\mathcal{C}_{K,V}$ for every $\mathbf{z}$. Hence
$$
\E_{\mathbf{z}}\E_{(\mathbf{x},y)\sim D}[(f(\mathbf{x})-h_{\mathbf{z}}(\mathbf{x}))y]
=
\E_{(\mathbf{x},y)\sim D}[(f(\mathbf{x})-f_{V}(\mathbf{x}))y]
\le
\eps.
$$
Therefore there exists some $\mathbf{z}$ such that
$
\E_{(\mathbf{x},y)\sim D}[(f(\mathbf{x})-h_{\mathbf{z}}(\mathbf{x}))y]\le \eps.
$
Since $h_{\mathbf{z}}\in\mathcal{C}_{K,V}$, taking $h=h_{\mathbf{z}}$ proves the consequent.
\end{proof}

\subsection{Finite Search in the Reduced Subspace for Boolean Functions of Halfspaces}

In this section, we show that the existence guarantee from \Cref{sec:existence-guarantee} can be converted into a proper learner for arbitrary Boolean functions of halfspaces. 

A brute-force search over all Boolean maps $B:\{-1,+1\}^K\to\{-1,+1\}$ requires considering $2^{2^K}$ candidates, which is doubly exponential in $K$. Instead, we use the standard hyperplane-arrangement trick, formalized in \Cref{lem:boolean-map-erm}: after the fixed halfspaces partition the sample into cells, the optimal Boolean map is obtained by optimizing independently on each cell.

\begin{lemma}[Empirical optimization over the Boolean map]\label{lem:boolean-map-erm}
Fix a sample $S$ and halfspaces $h_1,\dots,h_K$. For each $b\in\{\pm1\}^K$, let
$$
S_b:=\{(\mathbf{x},y)\in S:(h_1(\mathbf{x}),\dots,h_K(\mathbf{x}))=b\},
$$
and let $N_b^+$ and $N_b^-$ be the numbers of positive and negative labels in $S_b$. Then
$$
\min_{g:\{\pm1\}^K\to\{\pm1\}} \err_S(g(h_1,\dots,h_K))
=
\frac{1}{|S|}\sum_{b\in\{\pm1\}^K}\min\{N_b^+,N_b^-\}.
$$
The minimizer is obtained by assigning each cell its majority label, and can be computed in time $O(K|S|+2^K)$.
\end{lemma}

\begin{proof}
For fixed $h_1,\dots,h_K$, the cells $S_b$ partition the sample, and $g(h_1,\dots,h_K)$ is constant on each cell. Hence the empirical error separates over the cells. On cell $S_b$, the optimal choice of $g(b)$ is the majority label, contributing $\min\{N_b^+,N_b^-\}/|S|$ to the error. Summing over all $b\in\{\pm1\}^K$ gives the claim. Computing the cell counts requires evaluating the $K$ halfspaces on the sample and then assigning a majority label to each of $2^K$ cells.
\end{proof}

We now restate and prove the main result of our paper (\Cref{thm:main-B}).

\begin{theorem}[Proper learning of Boolean functions of halfspaces]\label{app:thm:main-B}
There exists a proper agnostic learner for the class of arbitrary Boolean functions of $K$ halfspaces under Gaussian marginals with the following guarantee. For every $K\ge 1$, every $\eps,\delta\in(0,1/2)$, and every distribution $D$ over $(\mathbf{x},y)\in\R^d\times\{\pm1\}$ whose $\mathbf{x}$-marginal is $\cN(0,\mathbf{I}_d)$, the learner draws
$
N
=
d^{O(K^2\log(1/\eps)/\eps^2)}\log(1/\delta)
$
samples, runs in time
$
d^{O(K^2\log(1/\eps)/\eps^2)}\log(1/\delta)
+
(K/\eps)^{O(K^3/\eps^{2.5})}\log(1/\delta),
$
and returns a hypothesis $h\in\mathcal{C}_{K}$ such that, with probability at least $1-\delta$,
$
\err_D(h)\le \OPT_K+\eps.
$
\end{theorem}

\begin{proof}
Set
$
k=\Theta\!\left(\frac{K^2\log(1/\eps)}{\eps^2}\right),
$
$
\eta=\Theta\!\left(\frac{\eps^2}{K}\right),
$
and
$
\nu=\Theta\!\left(\frac{\eps^{3/2}}{K^{3/2}}\right).
$

By \Cref{prop:compute-P}, using
$
N_1
=
d^{O(K^2\log(1/\eps)/\eps^2)}\log(1/\delta)
$
samples, we compute a polynomial $P\in\cP_k$ such that
$
\cL_{\mu,\nu}(P)\le \inf_{P'\in\cP_k}\cL_{\mu,\nu}(P')+O(\eps^3)
$
with probability at least $1-\delta/3$.

Form the influence matrix
$
\mathbf{M}(P)=\E_{\mathbf{x}\sim\cN(0,\mathbf{I}_d)}[\nabla P(\mathbf{x})\nabla P(\mathbf{x})^\top],
$
let $V$ be the span of the eigenvectors with eigenvalues at least $\eta$, and let $r=\dim(V)$. Although \Cref{lem:dim-bound} was stated in the halfspace section, it is true regardless of the target concept class and applies verbatim here.
$$
r
=
O\!\left(\frac{1}{\nu\sqrt{\eta}}\right)
=
O\!\left(\frac{K^2}{\eps^{2.5}}\right).
$$

Fix any $f\in\mathcal{C}_K$ with $\err_D(f)\le \OPT_K+\eps$. By \Cref{thm:structural-B}, there exists $h_V\in\mathcal{C}_{K,V}$ such that
$$
\E_{(\mathbf{x},y)\sim D}[(f(\mathbf{x})-h_V(\mathbf{x}))y]\le \eps.
$$
Therefore
$$
\err_D(h_V)
\le
\err_D(f)+\frac{\eps}{2}
\le
\OPT_K+\frac{3\eps}{2}.
$$

Let $\cH\subseteq\mathcal{C}_V$ be a halfspace cover at accuracy $c\eps/K$, for a sufficiently small absolute constant $c>0$. By \Cref{fact:cover},
$$
|\cH|=(K/\eps)^{O(r)}.
$$
Define the finite class
$$
\cG
:=
\left\{
\mathbf{x}\mapsto g(h_1(\mathbf{x}),\dots,h_K(\mathbf{x})):
h_1,\dots,h_K\in\cH,\; g:\{\pm1\}^K\to\{\pm1\}
\right\}.
$$
We claim that $\cG$ contains a hypothesis with error at most $\OPT_K+O(\eps)$.

Indeed, write
$
h_V(\mathbf{x})=B_0(g_1(\mathbf{x}),\dots,g_K(\mathbf{x}))
$
for some $g_1,\dots,g_K\in\mathcal{C}_V$ and some Boolean map $B_0:\{-1, +1\}^K\to\{-1, +1\}$. For each $i$, choose $\widetilde g_i\in\cH$ such that
$$
\Pr_{\mathbf{x}\sim\cN(0,\mathbf{I}_d)}[g_i(\mathbf{x})\neq \widetilde g_i(\mathbf{x})]\le c\eps/K.
$$
Then, by the union bound,
$$
\Pr_{\mathbf{x}\sim\cN(0,\mathbf{I}_d)}
[
B_0(g_1(\mathbf{x}),\dots,g_K(\mathbf{x}))
\neq
B_0(\widetilde g_1(\mathbf{x}),\dots,\widetilde g_K(\mathbf{x}))
]
\le
\sum_{i=1}^K
\Pr[g_i(\mathbf{x})\neq \widetilde g_i(\mathbf{x})]
\le
c\eps.
$$
Thus
$$
\err_D\big(\mathbf{x}\mapsto B_0(\widetilde g_1(\mathbf{x}),\dots,\widetilde g_K(\mathbf{x}))\big)
\le
\err_D(h_V)+c\eps
\le
\OPT_K+O(\eps).
$$
Hence
$$
\min_{h\in\cG}\err_D(h)\le \OPT_K+O(\eps).
$$

It remains to select a good hypothesis from $\cG$. Since
$
|\cG|\le |\cH|^K2^{2^K}
$
and $|\cH|=(K/\eps)^{O(r)}$, we have
$$
\log|\cG|\le O(Kr\log(K/\eps)+2^K).
$$
With $r=O(K^2/\eps^{2.5})$, draw an independent validation sample of size
$$
N_2
=
O\!\left(
\frac{Kr\log(K/\eps)+2^K+\log(1/\delta)}{\eps^2}
\right)
=
O\!\left(
\frac{K^3\log(K/\eps)}{\eps^{4.5}}
+
\frac{2^K+\log(1/\delta)}{\eps^2}
\right).
$$
By Hoeffding's inequality and a union bound over $\cG$, with probability at least $1-\delta/3$, empirical risk minimization over $\cG$ returns a hypothesis $h$ satisfying
$$
\err_D(h)
\le
\min_{h'\in\cG}\err_D(h')+\eps/2
\le
\OPT_K+O(\eps).
$$
Rescaling $\eps$ by a sufficiently small absolute constant gives the stated $\OPT_K+\eps$ guarantee.

Algorithmically, we do not enumerate all $2^{2^K}$ Boolean maps. We enumerate the tuples $(h_1,\dots,h_K)\in\cH^K$, and for each fixed tuple compute the empirically optimal Boolean map using \Cref{lem:boolean-map-erm}. The cost per tuple is $O(KN_2+2^K)$, so the total search time is bounded by
$$
|\cH|^K\cdot \operatorname{poly}(2^K,K,1/\eps,\log(1/\delta)).
$$
Since $|\cH|^K=(K/\eps)^{O(Kr)}$ and
$
r=O(K^2/\eps^{2.5}),
$
the factor $2^K$ is absorbed into
$$
(K/\eps)^{O(Kr)}
=
(K/\eps)^{O(K^3/\eps^{2.5})}.
$$
Combining this with the time for computing $P$ and $V$ gives the claimed running time.
\end{proof}

% We restate and prove the main result of our paper.

% \begin{theorem}[Proper learning of Boolean functions of halfspaces]\label{thm:main-B}
% There exists a proper agnostic learner for the class of arbitrary Boolean functions of $K$ halfspaces under Gaussian marginals with the following guarantee. For every $K\ge 1$, every $\eps,\delta\in(0,1/2)$, and every distribution $D$ on $\R^d\times\{\pm1\}$ whose $\mathbf{x}$-marginal is $\cN(0,\mathbf{I}_d)$, the learner draws
% $$
% N
% =
% d^{O(K^2\log(1/\eps)/\eps^2)}\log(1/\delta)
% $$
% samples, runs in time
% $$
% d^{O(K^2\log(1/\eps)/\eps^2)}\log(1/\delta)
% +
% (K/\eps)^{O(K^3/\eps^{2.5})}\log(1/\delta),
% $$
% and returns a hypothesis $h\in\mathcal{C}_{K}$ such that, with probability at least $1-\delta$,
% $$
% \err_D(h)\le \OPT_K+\eps.
% $$
% \end{theorem}

\section{Proper Learning of Intersections of Halfspaces}
\label{app:intersections}

Perhaps the most important special case of a Boolean function of $K$ halfspaces is an intersection of halfspaces, obtained by taking their conjunction. We denote this class by
$
\mathcal{C}_K^{\cap}
$,
and write $\mathcal{C}_{K,V}^{\cap}$ for the subclass in which all halfspace normal vectors lie in $V$. We also write
$$
\OPT_K^{\cap}:=\inf_{f\in\mathcal{C}_K^{\cap}}\err_D(f).
$$

Although the previous section already applies to arbitrary Boolean functions of halfspaces, intersections have additional geometric structure. In particular, Nazarov's Gaussian surface area bound gives a sharper approximation theorem, and this leads to a better dependence on $K$.

\begin{fact}[GSA bound for intersections of halfspaces {\normalfont\NazBound}]\label{fact:gsa-intersections}
There exists an absolute constant $C>0$ such that for every $K\ge 2$ and every $f\in\mathcal{C}_K^{\cap}$,
$$
\GSA(f)\le O(\sqrt{\log K}).
$$
\end{fact}

This improves the general bound $O(K)$ for arbitrary Boolean functions of $K$ halfspaces to $O(\sqrt{\log K})$ for intersections. Using the same Hermite-concentration argument as in \Cref{thm:boolean-approx}, we obtain the following sharper approximation theorem.

\begin{theorem}[Gaussian $L_1$ approximation for intersections of halfspaces]\label{thm:intersection-approx}
There exist absolute constants $C_{\mathrm{IHS}},C'_{\mathrm{IHS}}>0$ with the following property. For every integer $K\ge 2$, every $\tau\in(0,1/10)$, and every $f\in\mathcal{C}_K^{\cap}$, there exists a polynomial $S:\R^d\to\R$ of degree at most
$
C_{\mathrm{IHS}}\frac{\log K\,\log(1/\tau)}{\tau^2}
$
such that
$$
\E_{\mathbf{x}\sim\cN(0,\mathbf{I}_d)}[\,|f(\mathbf{x})-S(\mathbf{x})|\,]\le \tau,
$$
$$
\E_{\mathbf{x}\sim\cN(0,\mathbf{I}_d)}\|\nabla S(\mathbf{x})\|^2\le C'_{\mathrm{IHS}}\frac{\log K}{\tau},
$$
$$
\E_{\mathbf{x}\sim\cN(0,\mathbf{I}_d)}[S(\mathbf{x})^2]\le 1.
$$
\end{theorem}

\begin{proof}
The proof is identical to that of \Cref{thm:boolean-approx}, except that the Gaussian surface area estimate from \Cref{fact:gsa-boolean} is replaced by Nazarov's sharper bound from \Cref{fact:gsa-intersections}. Thus the factor $K^2$ in the degree and gradient bounds is replaced by $\log K$.
\end{proof}

We now state the corresponding low-dimensional subspace theorem. Its proof is the same as the proof of \Cref{thm:structural-B}, with the approximation theorem above replacing \Cref{thm:boolean-approx}.

\begin{theorem}[Existence of a Low-Dimensional Subspace for intersections of halfspaces]\label{thm:structural-I}
There exists an absolute constant $C>0$ and an absolute constant $c_\nu>0$ such that the following holds. Fix $K\ge 2$ and $\eps\in(0,1/10)$, and set
$
k:=C\frac{\log K\,\log(1/\eps)}{\eps^2},
$
$
\mu:=1/128,
$
$
\eta:=\frac{\eps^2}{1024K},
$
and
$
\nu:=c_\nu\frac{\eps^{3/2}}{\sqrt{K\log K}}.
$
Let $P\in\cP_k$ satisfy
$
\cL_{\mu,\nu}(P)\le \inf_{P'\in\cP_k}\cL_{\mu,\nu}(P')+O(\eps^3),
$
where the hidden constant is sufficiently small. Let $V$ be the span of the eigenvectors of $\mathbf{M}(P)$ whose eigenvalues are at least $\eta$. Then for every $f\in\mathcal{C}_K^{\cap}$, the averaged function $f_V$ from \Cref{def:boolean-averaging} satisfies
$$
\E_{(\mathbf{x},y)\sim D}[(f(\mathbf{x})-f_V(\mathbf{x}))y]\le \eps.
$$
Consequently, there exists $h\in\mathcal{C}_{K,V}^\cap$ such that
$$
\E_{(\mathbf{x},y)\sim D}[(f(\mathbf{x})-h(\mathbf{x}))y]\le \eps.
$$
\end{theorem}

\begin{proof}
The proof is identical to that of \Cref{thm:structural-B}, replacing \Cref{thm:boolean-approx} by \Cref{thm:intersection-approx}. The only changes are 
$
k=O\!\left(\frac{\log K\,\log(1/\eps)}{\eps^2}\right)
$
and
$
\tr(\mathbf{M}(Q)^{1/2})
\le
O\!\left(\sqrt{\frac{K\log K}{\tau}}\right),
$
instead of
$
O(K^{3/2}/\sqrt{\tau})
$.
Accordingly, choosing
$
\nu=c_\nu\eps^{3/2}/\sqrt{K\log K}
$
with $c_\nu>0$ sufficiently small makes the regularization term in the gain bound at most a small constant multiple of $\eps$.
\end{proof}

We now turn the structural theorem into a proper learner. Since the Boolean map is fixed to be conjunction, there is no search over arbitrary Boolean maps.

\begin{theorem}[Proper learning of intersections of halfspaces]\label{thm:main-I}
There exists a proper agnostic learner for $\mathcal{C}_K^{\cap}$ under Gaussian marginals with the following guarantee. For every $K\ge 2$, every $\eps,\delta\in(0,1/2)$, and every distribution $D$ over $(\mathbf{x},y)\in\R^d\times\{\pm1\}$ whose $\mathbf{x}$-marginal is $\cN(0,\mathbf{I}_d)$, the learner draws
$
N
=
d^{O(\log K\,\log(1/\eps)/\eps^2)}\log(1/\delta)
$
samples, runs in time
$
d^{O(\log K\,\log(1/\eps)/\eps^2)}\log(1/\delta)
+
\left(\frac{\sqrt{\log K}}{\eps}\right)^{O(K^2\sqrt{\log K}/\eps^{2.5})}\log(1/\delta),
$
and returns a hypothesis $h\in\mathcal{C}_K^{\cap}$ such that, with probability at least $1-\delta$,
$
\err_D(h)\le \OPT_K^{\cap}+\eps.
$
\end{theorem}

\begin{proof}
The proof is identical to the proof of \Cref{thm:main-B}, replacing \Cref{thm:structural-B} by \Cref{thm:structural-I} and replacing the finite class $\cG$ by the cover of intersections inside $V$. With
$
\nu=\Theta(\eps^{3/2}/\sqrt{K\log K})
$
and
$
\eta=\Theta(\eps^2/K),
$
\Cref{lem:dim-bound} gives
$$
\dim(V)
=
O\!\left(\frac{1}{\nu\sqrt{\eta}}\right)
=
O\!\left(\frac{K\sqrt{\log K}}{\eps^{2.5}}\right).
$$
A $\eps$-cover of intersections of $K$ halfspaces in $V$ has size
$$
\left(\frac{\sqrt{\log K}}{\eps}\right)^{O(K\dim(V))}
=
\left(\frac{\sqrt{\log K}}{\eps}\right)^{O(K^2\sqrt{\log K}/\eps^{2.5})}.
$$
Empirical risk minimization over this cover, using an independent validation sample and a union bound, returns a hypothesis with error at most $\OPT_K^{\cap}+\eps$. The running time is the sum of the time needed to compute $P$ and $V$, and the time needed to search over the cover, which gives the claimed bound.
\end{proof}

\bibliographystyle{alpha}
\bibliography{refs}

@book{Riv74,
  author    = {Rivlin, T. J.},
  title     = {The Chebyshev Polynomials},
  publisher = {Wiley},
  year      = {1974},
}

@article{Berkson1944,
  author = {Berkson, Joseph},
  title = {Application of the Logistic Function to Bio-Assay},
  journal = {Journal of the American Statistical Association},
  volume = {39},
  number = {227},
  pages = {357--365},
  year = {1944}
}

@article{Cox1958,
  author = {Cox, D. R.},
  title = {The Regression Analysis of Binary Sequences},
  journal = {Journal of the Royal Statistical Society: Series B},
  volume = {20},
  number = {2},
  pages = {215--242},
  year = {1958}
}

@article{CortesVapnik1995,
  author = {Cortes, Corinna and Vapnik, Vladimir},
  title = {Support-Vector Networks},
  journal = {Machine Learning},
  volume = {20},
  number = {3},
  pages = {273--297},
  year = {1995}
}

@article{FriedmanHastieTibshirani2000,
  author = {Friedman, Jerome and Hastie, Trevor and Tibshirani, Robert},
  title = {Additive Logistic Regression: A Statistical View of Boosting},
  journal = {The Annals of Statistics},
  volume = {28},
  number = {2},
  pages = {337--407},
  year = {2000}
}

@article{BartlettJordanMcAuliffe2006,
  author = {Bartlett, Peter L. and Jordan, Michael I. and McAuliffe, Jon D.},
  title = {Convexity, Classification, and Risk Bounds},
  journal = {Journal of the American Statistical Association},
  volume = {101},
  number = {473},
  pages = {138--156},
  year = {2006}
}

@article{HoerlKennard1970,
  author = {Hoerl, Arthur E. and Kennard, Robert W.},
  title = {Ridge Regression: Biased Estimation for Nonorthogonal Problems},
  journal = {Technometrics},
  volume = {12},
  number = {1},
  pages = {55--67},
  year = {1970}
}

@book{TikhonovArsenin1977,
  author = {Tikhonov, A. N. and Arsenin, V. Y.},
  title = {Solutions of Ill-Posed Problems},
  publisher = {Winston},
  year = {1977}
}

@article{BartlettMendelson2002,
  author = {Bartlett, Peter L. and Mendelson, Shahar},
  title = {Rademacher and Gaussian Complexities: Risk Bounds and Structural Results},
  journal = {Journal of Machine Learning Research},
  volume = {3},
  pages = {463--482},
  year = {2002}
}

@book{ScholkopfSmola2002,
  author = {Sch{\"o}lkopf, Bernhard and Smola, Alexander J.},
  title = {Learning with Kernels},
  publisher = {MIT Press},
  year = {2002}
}

@article{Tibshirani1996,
  author = {Tibshirani, Robert},
  title = {Regression Shrinkage and Selection via the Lasso},
  journal = {Journal of the Royal Statistical Society: Series B},
  volume = {58},
  number = {1},
  pages = {267--288},
  year = {1996}
}

@inproceedings{SrebroRennieJaakkola2004,
  author = {Srebro, Nathan and Rennie, Jason D. M. and Jaakkola, Tommi S.},
  title = {Maximum-Margin Matrix Factorization},
  booktitle = {Advances in Neural Information Processing Systems 17},
  pages = {1329--1336},
  year = {2004}
}

@article{RechtFazelParrilo2010,
  author = {Recht, Benjamin and Fazel, Maryam and Parrilo, Pablo A.},
  title = {Guaranteed Minimum-Rank Solutions of Linear Matrix Equations via Nuclear Norm Minimization},
  journal = {SIAM Review},
  volume = {52},
  number = {3},
  pages = {471--501},
  year = {2010}
}

@inproceedings{DiakonikolasKaneKontonisTzamosZarifis2021,
  author = {Diakonikolas, Ilias and Kane, Daniel M. and Kontonis, Vasilis and Tzamos, Christos and Zarifis, Nikos},
  title = {Agnostic Proper Learning of Halfspaces under Gaussian Marginals},
  booktitle = {Proceedings of the 34th Conference on Learning Theory},
  series = {Proceedings of Machine Learning Research},
  volume = {134},
  pages = {1522--1551},
  year = {2021}
}

@inproceedings{DKN10,
  author    = {Diakonikolas, I. and Kane, D. M. and Nelson, J.},
  title     = {Bounded independence fools degree-2 threshold functions},
  booktitle = {Proceedings of the 51st IEEE Symposium on Foundations of Computer Science (FOCS)},
  pages     = {11--20},
  year      = {2010},
}

@book{szego1967orthogonal,
  author    = {Szeg\H{o}, G{\'a}bor},
  title     = {Orthogonal Polynomials},
  series    = {American Mathematical Society Colloquium Publications},
  volume    = {23},
  publisher = {American Mathematical Society},
  year      = {1967},
}

@article{KKMS08,
  author  = {Kalai, Adam Tauman and Klivans, Adam R. and Mansour, Yishay and Servedio, Rocco A.},
  title   = {Agnostically Learning Halfspaces},
  journal = {SIAM Journal on Computing},
  volume  = {37},
  number  = {6},
  pages   = {1777--1805},
  year    = {2008}
}

@misc{DKPZ21,
  author = {Diakonikolas, I. and Kane, D. M. and Pittas, T. and Zarifis, N.},
  title = {The Optimality of Polynomial Regression for Agnostic Learning under Gaussian Marginals},
  year = {2021},
  note = {arXiv:2106.07131}
}

@inproceedings{DKTKZ21,
  shorthand = {DKTKZ21},
  author    = {Diakonikolas, I. and Kane, D. M. and Kontonis, V. and Tzamos, C. and Zarifis, N.},
  title     = {Agnostic proper learning of halfspaces under Gaussian marginals},
  booktitle = {Proceedings of the 34th Conference on Learning Theory},
  series    = {Proceedings of Machine Learning Research},
  volume    = {134},
  pages     = {1522--1551},
  year      = {2021},
}

@inproceedings{KOS08,
  author    = {Klivans, A. R. and O'Donnell, R. and Servedio, R. A.},
  title     = {Learning geometric concepts via Gaussian surface area},
  booktitle = {49th Annual IEEE Symposium on Foundations of Computer Science (FOCS)},
  pages     = {541--550},
  year      = {2008},
}

@book{Ver18,
  author    = {Vershynin, Roman},
  title     = {High-Dimensional Probability: An Introduction with Applications in Data Science},
  series    = {Cambridge Series in Statistical and Probabilistic Mathematics},
  publisher = {Cambridge University Press},
  year      = {2018},
}

@incollection{Naz03,
  author    = {Nazarov, F.},
  title     = {On the maximal perimeter of a convex set in {$\mathbb{R}^n$} with respect to a Gaussian measure},
  booktitle = {Geometric Aspects of Functional Analysis},
  series    = {Lecture Notes in Mathematics},
  volume    = {1807},
  pages     = {169--187},
  publisher = {Springer},
  year      = {2003},
}

@misc{PSW26,
  author = {Pesenti, L. and Slot, L. and Wiedmer, M.},
  title = {Agnostic Learning in (Almost) Optimal Time via Gaussian Surface Area},
  year = {2026},
  note = {arXiv:2603.06027}
}

@book{BLM13,
  author    = {Boucheron, S. and Lugosi, G. and Massart, P.},
  title     = {Concentration Inequalities: A Nonasymptotic Theory of Independence},
  publisher = {Oxford University Press},
  year      = {2013},
}

@book{OD14,
  author    = {O'Donnell, R.},
  title     = {Analysis of Boolean Functions},
  publisher = {Cambridge University Press},
  year      = {2014},
}

@inproceedings{lange2022properly,
  title={Properly learning monotone functions via local correction},
  author={Lange, Jane and Rubinfeld, Ronitt and Vasilyan, Arsen},
  booktitle={2022 IEEE 63rd Annual Symposium on Foundations of Computer Science (FOCS)},
  pages={75--86},
  year={2022},
  organization={IEEE}
}

@article{lange2025agnostic,
  title={Agnostic proper learning of monotone functions: beyond the black-box correction barrier},
  author={Lange, Jane and Vasilyan, Arsen},
  journal={SIAM Journal on Computing},
  pages={FOCS23--1},
  year={2025},
  publisher={SIAM}
}

@article{pinto2025learning,
  title={Learning and Testing Convex Functions},
  author={Pinto Jr, Renato Ferreira and Marcussen, Cassandra and Mossel, Elchanan and Nadimpalli, Shivam},
  journal={arXiv preprint arXiv:2511.11498},
  year={2025}
}

@article{mehta2002decision,
  author  = {Mehta, Dinesh P. and Raghavan, Vijay},
  title   = {Decision Tree Approximations of Boolean Functions},
  journal = {Theoretical Computer Science},
  volume  = {270},
  number  = {1--2},
  pages   = {609--623},
  year    = {2002},
  doi     = {10.1016/S0304-3975(01)00011-1}
}

@article{awasthi2017power,
  title={The power of localization for efficiently learning linear separators with noise},
  author={Awasthi, Pranjal and Balcan, Maria Florina and Long, Philip M},
  journal={Journal of the ACM (JACM)},
  volume={63},
  number={6},
  pages={1--27},
  year={2017},
  publisher={ACM New York, NY, USA}
}

@inproceedings{diakonikolas2018learning,
  title={Learning geometric concepts with nasty noise},
  author={Diakonikolas, Ilias and Kane, Daniel M and Stewart, Alistair},
  booktitle={Proceedings of the 50th Annual ACM SIGACT Symposium on Theory of Computing},
  pages={1061--1073},
  year={2018}
}

@inproceedings{diakonikolas2022learning,
  title={Learning general halfspaces with adversarial label noise via online gradient descent},
  author={Diakonikolas, Ilias and Kontonis, Vasilis and Tzamos, Christos and Zarifis, Nikos},
  booktitle={International Conference on Machine Learning},
  pages={5118--5141},
  year={2022},
  organization={PMLR}
}

@article{diakonikolas2020non,
  title={Non-convex SGD learns halfspaces with adversarial label noise},
  author={Diakonikolas, Ilias and Kontonis, Vasilis and Tzamos, Christos and Zarifis, Nikos},
  journal={Advances in Neural Information Processing Systems},
  volume={33},
  pages={18540--18549},
  year={2020}
}

@inproceedings{vempala2010corrigendum,
  author    = {Vempala, Santosh S.},
  title     = {Corrigendum: A Random Sampling Algorithm for Learning an Intersection of Halfspaces},
  booktitle = {Proceedings of the 2010 IEEE 51st Annual Symposium on Foundations of Computer Science},
  series    = {FOCS '10},
  pages     = {123},
  year      = {2010},
  publisher = {IEEE Computer Society},
  doi       = {10.1109/FOCS.2010.18}
}

@article{vempala2010randomsampling,
  author  = {Vempala, Santosh S.},
  title   = {A Random-Sampling-Based Algorithm for Learning Intersections of Halfspaces},
  journal = {Journal of the ACM},
  volume  = {57},
  number  = {6},
  articleno = {32},
  pages   = {32:1--32:14},
  month   = nov,
  year    = {2010},
  doi     = {10.1145/1857914.1857916},
  publisher = {Association for Computing Machinery}
}

@inproceedings{vempala2010learningconvex,
  author    = {Vempala, Santosh S.},
  title     = {Learning Convex Concepts from Gaussian Distributions with {PCA}},
  booktitle = {Proceedings of the 2010 IEEE 51st Annual Symposium on Foundations of Computer Science},
  series    = {FOCS '10},
  pages     = {124--130},
  year      = {2010},
  publisher = {IEEE Computer Society},
  doi       = {10.1109/FOCS.2010.19}
}

@article{ehrenfeucht1989learning,
  author  = {Ehrenfeucht, Andrzej and Haussler, David},
  title   = {Learning Decision Trees from Random Examples},
  journal = {Information and Computation},
  volume  = {82},
  number  = {3},
  pages   = {231--246},
  year    = {1989},
  doi     = {10.1016/0890-5401(89)90001-1}
}

@article{blanc2022properly,
  title={Properly learning decision trees in almost polynomial time},
  author={Blanc, Guy and Lange, Jane and Qiao, Mingda and Tan, Li-Yang},
  journal={Journal of the ACM},
  volume={69},
  number={6},
  pages={1--19},
  year={2022},
  publisher={ACM New York, NY}
}

@inproceedings{diakonikolas2024agnostically,
  author    = {Diakonikolas, Ilias and Kane, Daniel M. and Kontonis, Vasilis and Tzamos, Christos and Zarifis, Nikos},
  title     = {Agnostically Learning Multi-Index Models with Queries},
  booktitle = {2024 IEEE 65th Annual Symposium on Foundations of Computer Science (FOCS)},
  pages     = {1931--1952},
  year      = {2024},
  publisher = {IEEE},
  doi       = {10.1109/FOCS61266.2024.00115}
}

@InProceedings{pmlr-v134-diakonikolas21c,
  title     = {The Optimality of Polynomial Regression for Agnostic Learning under Gaussian Marginals in the SQ Model},
  author    = {Diakonikolas, Ilias and Kane, Daniel M. and Pittas, Thanasis and Zarifis, Nikos},
  booktitle = {Proceedings of Thirty Fourth Conference on Learning Theory},
  pages     = {1552--1584},
  year      = {2021},
  volume    = {134},
  series    = {Proceedings of Machine Learning Research},
  publisher = {PMLR}
}

@article{AlbertAnderson1984,
  author  = {Albert, A. and Anderson, J. A.},
  title   = {On the Existence of Maximum Likelihood Estimates in Logistic Regression Models},
  journal = {Biometrika},
  volume  = {71},
  number  = {1},
  pages   = {1--10},
  year    = {1984},
  doi     = {10.1093/biomet/71.1.1}
}

@article{Donoho2006,
  author  = {David L. Donoho},
  title   = {Compressed Sensing},
  journal = {IEEE Transactions on Information Theory},
  volume  = {52},
  number  = {4},
  pages   = {1289--1306},
  year    = {2006},
  doi     = {10.1109/TIT.2006.871582}
}

@inproceedings{HsuSSV22,
  author    = {Hsu, Daniel and Sanford, Clayton and Servedio, Rocco A. and Vlatakis-Gkaragkounis, Emmanouil V.},
  title     = {Near-Optimal Statistical Query Lower Bounds for Agnostically Learning Intersections of Halfspaces with {G}aussian Marginals},
  booktitle = {Proceedings of the 35th Conference on Learning Theory ({COLT})},
  series    = {Proceedings of Machine Learning Research},
  volume    = {178},
  pages     = {283--312},
  year      = {2022},
  publisher = {PMLR},
  eprint    = {2202.05096},
  eprinttype = {arXiv},
}

@inproceedings{lange2025robust,
  title     = {Robust Learning of Halfspaces Under Log-Concave Marginals},
  author    = {Lange, Jane and Vasilyan, Arsen},
  booktitle = {Advances in Neural Information Processing Systems},
  year      = {2025}
}

@article{rosenblatt1958perceptron,
  title     = {The Perceptron: A Probabilistic Model for Information Storage and Organization in the Brain},
  author    = {Rosenblatt, Frank},
  journal   = {Psychological Review},
  volume    = {65},
  number    = {6},
  pages     = {386--408},
  year      = {1958},
  publisher = {American Psychological Association},
  doi       = {10.1037/h0042519}
}

@InProceedings{pmlr-v202-diakonikolas23b,
  title     = {Near-Optimal Cryptographic Hardness of Agnostically Learning Halfspaces and {R}e{LU} Regression under {G}aussian Marginals},
  author    = {Diakonikolas, Ilias and Kane, Daniel and Ren, Lisheng},
  booktitle = {Proceedings of the 40th International Conference on Machine Learning},
  pages     = {7922--7938},
  year      = {2023},
  volume    = {202},
  month     = {23--29 Jul},
  publisher = {PMLR},
  url       = {https://proceedings.mlr.press/v202/diakonikolas23b.html}
}

@article{diakonikolas2020near,
  title     = {Near-optimal {SQ} lower bounds for agnostically learning halfspaces and {ReLUs} under {Gaussian} marginals},
  author    = {Diakonikolas, Ilias and Kane, Daniel and Zarifis, Nikos},
  journal   = {Advances in Neural Information Processing Systems},
  volume    = {33},
  pages     = {13586--13596},
  year      = {2020}
}

@article{goel2020statistical,
  title   = {Statistical-query lower bounds via functional gradients},
  author  = {Goel, Surbhi and Gollakota, Aravind and Klivans, Adam},
  journal = {Advances in Neural Information Processing Systems},
  volume  = {33},
  pages   = {2147--2158},
  year    = {2020}
}

@inproceedings{guruswami2006hardness,
  title     = {Hardness of learning halfspaces with noise},
  author    = {Guruswami, Venkatesan and Raghavendra, Prasad},
  booktitle = {Proceedings of the 47th IEEE Symposium on Foundations of Computer Science (FOCS)},
  pages     = {543--552},
  year      = {2006},
  publisher = {IEEE Computer Society}
}

@inproceedings{feldman2006new,
  title     = {New results for learning noisy parities and halfspaces},
  author    = {Feldman, Vitaly and Gopalan, Parikshit and Khot, Subhash and Ponnuswami, Ashok Kumar},
  booktitle = {Proceedings of the 47th Annual IEEE Symposium on Foundations of Computer Science (FOCS)},
  year      = {2006},
  publisher = {IEEE Computer Society}
}

@inproceedings{alman2026learning,
  title     = {Learning Functions of Halfspaces},
  author    = {Alman, Josh and Patel, Shyamal and Servedio, Rocco A.},
  booktitle = {Proceedings of the 58th Annual ACM Symposium on Theory of Computing (STOC)},
  year      = {2026},
  note      = {To appear}
}

\appendix

\section{Hermite Polynomials}

For a comprehensive discussion of Hermite polynomials, we refer the reader to Section~11.2 of \cite{OD14} and Chapter~5 of \cite{szego1967orthogonal}.

\begin{definition}[Normalized Hermite basis]
For $j\in\mathbb N$, let $H_j:\R\to\R$ denote the normalized univariate Hermite polynomial
$$
H_j(t)
=
\frac{(-1)^j}{\sqrt{j!}}
\exp(t^2/2)\frac{d^j}{dt^j}\exp(-t^2/2).
$$
For a multi-index $\alpha=(\alpha_1,\dots,\alpha_d)\in\mathbb N^d$, write $|\alpha|:=\sum_{i=1}^d \alpha_i$ and
$$
H_\alpha(\mathbf{x})
:=
\prod_{i=1}^d H_{\alpha_i}(\mathbf{x}_i).
$$
We refer to $\{H_\alpha\}_{\alpha\in\mathbb N^d}$ as the normalized multivariate Hermite basis.
\end{definition}

\begin{fact}[Orthonormality and Parseval]\label{fact:hermite-parseval}
The basis $\{H_\alpha\}_{\alpha\in\mathbb N^d}$ is orthonormal in $L_2(\cN(0,\mathbf{I}_d))$:
$$
\E_{\mathbf{x}\sim\cN(0,\mathbf{I}_d)}[H_\alpha(\mathbf{x})H_\beta(\mathbf{x})]
=
\Ind {\{\alpha=\beta\}}.
$$
Consequently, every $P\in\cP_k$ has a unique expansion\footnote{We use the coefficient representation or the feature-map representation depending on convenience.}
$$
P(\mathbf{x}) = \sum_{|\alpha|\le k}c_\alpha H_\alpha(\mathbf{x}) = \langle \mathbf{c},\Phi(\mathbf{x})\rangle
$$
Parseval's identity gives
$$
\E_{\mathbf{x}\sim\cN(0,\mathbf{I}_d)}[P(\mathbf{x})^2]
=
\sum_{|\alpha|\le k}c_\alpha^2.
$$
\end{fact}

\begin{fact}[Derivative of Hermite Polynomial]\label{fact:hermite-derivative}
For every multi-index $\alpha\in\mathbb N^d$ and every $i\in[d]$,
$$
\partial_i H_\alpha(\mathbf{x})
=
\begin{cases}
\sqrt{\alpha_i}\,H_{\alpha-\mathbf{e}_i}(\mathbf{x}), & \alpha_i\ge 1,\\
0, & \alpha_i=0.
\end{cases}
$$
Consequently, if $P(\mathbf{x})=\sum_{|\alpha|\le k}c_\alpha H_\alpha(\mathbf{x})$, then
$$
\E_{\mathbf{x}\sim\cN(0,\mathbf{I}_d)}\|\nabla P(\mathbf{x})\|^2
=
\sum_{|\alpha|\le k}|\alpha|c_\alpha^2.
$$
\end{fact}

\begin{lemma}[Hermite expansion of the gradient]\label{lem:gradient-hermite}
Let $P(\mathbf{x})=\sum_{|\alpha|\le k}c_\alpha H_\alpha(\mathbf{x})$. Then
$$
\nabla P(\mathbf{x})
=
\sum_{|\beta|\le k-1}\mathbf{a}_\beta(P)H_\beta(\mathbf{x}),
$$
where $
\mathbf{a}_\beta(P) = \left(
\sqrt{\beta_1+1}\,c_{\beta+\mathbf{e}_1},
\ldots,
\sqrt{\beta_d+1}\,c_{\beta+\mathbf{e}_d}
\right)^\top$.
\end{lemma}

\begin{proof}
For every $i\in[d]$, differentiating the normalized Hermite polynomial coordinatewise gives
$$
\partial_i H_\alpha(\mathbf{x})=
\begin{cases}
\sqrt{\alpha_i}\,H_{\alpha-\mathbf{e}_i}(\mathbf{x}), & \alpha_i\ge 1,\\
0, & \alpha_i=0
\end{cases}
$$
Therefore
$$
\partial_i P(\mathbf{x})
=
\sum_{\substack{|\alpha|\le k\\ \alpha_i\ge 1}}
\sqrt{\alpha_i}\,c_\alpha H_{\alpha-\mathbf{e}_i}(\mathbf{x})
$$
Reindexing with $\beta=\alpha-\mathbf{e}_i$ and defining $a_{i,\beta}(P):=\sqrt{\beta_i+1}\,c_{\beta+\mathbf{e}_i}$, we get
$$
\partial_i P(\mathbf{x})=\sum_{|\beta|\le k-1}\sqrt{\beta_i+1}\,c_{\beta+\mathbf{e}_i}H_\beta(\mathbf{x})=\sum_{|\beta|\le k-1}a_{i,\beta}(P)H_\beta(\mathbf{x})
$$
Thus the coefficient of $H_\beta$ in the Hermite expansion of $\nabla P$ is precisely
$$
\mathbf{a}_\beta(P)
=
\left(
\sqrt{\beta_1+1}\,c_{\beta+\mathbf{e}_1},
\ldots,
\sqrt{\beta_d+1}\,c_{\beta+\mathbf{e}_d}
\right)^\top.
$$
\end{proof}

\section{The Nuclear Norm}\label{app:nuclear}

In this appendix we collect the main facts about the regularizer
$$
\tr\!\big(\mathbf{M}(P)^{1/2}\big) = \tr((\E[\nabla P(\mathbf{x}) \nabla P(\mathbf{x})^\top])^{1/2}).
$$
In particular, we show that it can naturally be viewed as the nuclear norm of the coefficient matrix of $\nabla P$ in the Hermite basis.

\begin{definition}[Nuclear norm]\label{app:nuclear-def}
Let $\mathbf{A}\in\R^{m\times n}$ have singular values $\sigma_1,\dots,\sigma_r$, where $r=\rank(\mathbf{A})$. Let $\lambda_1,\dots,\lambda_r$ be the nonzero eigenvalues of $\mathbf{A}^\top \mathbf{A}$ (equivalently, of $\mathbf{A}\mathbf{A}^\top$). The nuclear norm of $\mathbf{A}$ is defined by
$$
\|\mathbf{A}\|_*:=\sum_{i=1}^r \sigma_i = \sum_{i=1}^r \sqrt{\lambda_i}.
$$
Equivalently,
\begin{equation}\label{app:nuclear-def-eq}
    \|\mathbf{A}\|_*=\tr\left((\mathbf{A}^\top \mathbf{A})^{1/2}\right)
=\tr\left((\mathbf{A}\mathbf{A}^\top)^{1/2}\right).
\end{equation}
\end{definition}

The reason we are interested in this regularizer is that the properties of Hermite polynomials allow us to rewrite $\tr(\mathbf{M}(P)^{1/2})$ as the nuclear norm of a matrix constructed from the Hermite coefficients of the vector-valued polynomial $\nabla P$.

\begin{lemma}[Nuclear norm representation of the spectral regularizer]
\label{app:lem:nuclear-hermite}
Let $\mathbf{a}_\beta(P)\in\R^d$ be the Hermite gradient coefficients from \Cref{lem:gradient-hermite}. Define $\mathbf{A}(P)$ by collecting the vectors $\mathbf{a}_\beta(P)$ as columns:
$$
\mathbf{A}(P):=\big(\mathbf{a}_\beta(P)\big)_{|\beta|\le k-1}
$$
Then $\mathbf{M}(P)=\mathbf{A}(P)\mathbf{A}(P)^\top$ and consequently
$$
\tr(\mathbf{M}(P)^{1/2})
=
\tr\big((\mathbf{A}(P)\mathbf{A}(P)^\top)^{1/2}\big)
=
\|\mathbf{A}(P)\|_*
$$
\end{lemma}

\begin{proof}
By definition, $\mathbf{M}(P)=\E[\nabla P(\mathbf{x})\nabla P(\mathbf{x})^\top]$. By \Cref{lem:gradient-hermite},
$$
\mathbf{M}(P)
=
\E\left[
\left(\sum_\beta \mathbf{a}_\beta(P)H_\beta(\mathbf{x})\right)
\left(\sum_\gamma \mathbf{a}_\gamma(P)H_\gamma(\mathbf{x})\right)^\top
\right]
$$
Using orthonormality of the Hermite basis, \Cref{fact:hermite-parseval}, this becomes
$$
\mathbf{M}(P)
=
\sum_{|\beta|\le k-1}\mathbf{a}_\beta(P)\mathbf{a}_\beta(P)^\top
=
\mathbf{A}(P)\mathbf{A}(P)^\top
$$
Finally, by the definition of nuclear norm \eqref{app:nuclear-def-eq},
$$
\tr(\mathbf{M}(P)^{1/2})
=
\tr\big((\mathbf{A}(P)\mathbf{A}(P)^\top)^{1/2}\big)
=
\|\mathbf{A}(P)\|_*
$$

\end{proof}

As a result, $P\mapsto \tr(\mathbf{M}(P)^{1/2})$ is convex on $\cP_k$, since $P\mapsto \mathbf{A}(P)$ is linear and the nuclear norm is convex. Therefore, $\tr(\mathbf{M}(P)^{1/2})$ can be used as a regularizer in our convex optimization problem. 

\begin{lemma}[Rank-trace bound]
\label{lem:nuclear-rank-trace}
For every polynomial $P\in\cP_k$,
$$
\tr(\mathbf{M}(P)^{1/2})
\le
\sqrt{\rank(\mathbf{M}(P))\,\tr(\mathbf{M}(P))} = \sqrt{\rank(\mathbf{M}(P))\,\E[\|\nabla P(\mathbf{x})\|_2^2]}.
$$

\end{lemma}

\begin{proof}
Let $\lambda_1,\dots,\lambda_r$ be the nonzero eigenvalues of $\mathbf{M}(P)$, where $r:=\rank(\mathbf{M}(P))$. Since $\mathbf{M}(P)\succeq 0$, all $\lambda_i$ are nonnegative. Therefore
$$
\tr(\mathbf{M}(P)^{1/2})=\sum_{i=1}^r \sqrt{\lambda_i}.
$$
By Cauchy--Schwarz inequality,
$$
\left(\sum_{i=1}^r \sqrt{\lambda_i}\right)^2
\le
\left(\sum_{i=1}^r 1\right)\left(\sum_{i=1}^r \lambda_i\right)
=
r\,\tr(\mathbf{M}(P)) = \rank(\mathbf{M}(P))\,\tr(\mathbf{M}(P)).
$$
Taking square roots gives
$$
\tr(\mathbf{M}(P)^{1/2})
\le
\sqrt{\rank(\mathbf{M}(P))\,\tr(\mathbf{M}(P))}.
$$
For the second part, using linearity and the cyclic property of trace,
$$
\tr(\mathbf{M}(P))= \tr(\E[\nabla P(\mathbf{x}) \nabla P(\mathbf{x})^\top])= \E[ \nabla P(\mathbf{x})^\top\nabla P(\mathbf{x})]=\E[\|\nabla P(\mathbf{x})\|_2^2].
$$
\end{proof}

\section{Regularized Logistic-Loss Polynomial Regression}\label{app:compute_reg}

We restate and prove the following theorem. The proof is a  uniform convergence result for a Lipschitz loss over a bounded finite-dimensional class. The main technical difficulty is that, in the Hermite basis, the feature map $\Phi(\mathbf{x})$ is unbounded under the Gaussian distribution. To address this, we use a standard truncation argument.

\begin{theorem}[Regularized logistic-loss polynomial regression]\label{thm:compute-reg}
Fix $\mu>0$ to be an absolute constant and let $\nu\ge 0$. Let $D$ be a distribution over $(\mathbf{x},y)\in\R^d\times\{\pm1\}$ whose $\mathbf{x}$-marginal is $\cN(0,\mathbf{I}_d)$. Let $k\in\mathbb Z_+$ and $\eps,\delta>0$. There is an algorithm that draws
$
N=(dk)^{O(k)}\frac{\log(1/\delta)}{\eps^2}
$
samples from $D$, runs in time $\poly(N,d)$, and outputs a polynomial $P\in\cP_k$ such that
$
\cL_{\mu,\nu}(P)\le \inf_{P'\in\cP_k}\cL_{\mu,\nu}(P')+\eps
$
with probability at least $1-\delta$.
\end{theorem}

\begin{proof}
Let $\Phi(\mathbf{x}):=(H_\alpha(\mathbf{x}))\in\R^m$ be the normalized Hermite feature map with $m:=|\{\alpha\in\mathbb N^d:\ |\alpha|\le k\}|=\binom{d+k}{k}$. Every $P\in\cP_k$ has a unique coefficient vector $\mathbf{c}\in\R^m$ such that
$$
P_{\mathbf{c}}(\mathbf{x})=\langle \mathbf{c},\Phi(\mathbf{x})\rangle = \sum_{|\alpha|\le k}c_\alpha H_\alpha(\mathbf{x})
$$
Based on orthogonality of the Hermite basis (\Cref{fact:hermite-parseval}),
$$
\E[\Phi(\mathbf{x})\Phi(\mathbf{x})^\top]=\mathbf{I}_m,
\qquad
\E\|\Phi(\mathbf{x})\|_2^2=m,
\qquad
\E[P_{\mathbf{c}}(\mathbf{x})^2]=\|\mathbf{c}\|_2^2.
$$
By the nuclear norm representation of the trace in \Cref{app:lem:nuclear-hermite}, where $\mathbf{A}(\mathbf{c})$ is a fixed linear function of $\mathbf{c}$, $ \tr(\mathbf{M}(P_{\mathbf{c}})^{1/2})=\|\mathbf{A}(\mathbf{c})\|_*$. Thus, in the Hermite basis, our loss $\eqref{logistic-loss}$ is
$$
\cL_{\mu,\nu}(\mathbf{c})
=
\E_{(\mathbf{x},y)\sim D}\ell(y\langle \mathbf{c},\Phi(\mathbf{x})\rangle)
+
\mu\|\mathbf{c}\|_2^2
+
\nu\|\mathbf{A}(\mathbf{c})\|_*.
$$

Let $\ell_0:=\ell(0)=\log 2$ and let $\mathcal B_R$ be the Euclidean ball of radius $R:=\sqrt{\ell_0/\mu}$:
$$
\mathcal B_R:=\{\mathbf{c}\in\R^m:\|\mathbf{c}\|_2\le R\}.
$$
Since $\ell\ge 0$ and $\|\mathbf{A}(\mathbf{c})\|_*\ge 0$, we have
$
\cL_{\mu,\nu}(\mathbf{c})\ge \mu\|\mathbf{c}\|_2^2
$,
while $\cL_{\mu,\nu}(0)=\ell_0$. Hence every population minimizer lies in $\mathcal B_R$. Since $\mu$ is an absolute constant, $R=O(1)$.

The technical difficulty is that the Hermite feature map $\Phi(\mathbf{x})$ is unbounded under the Gaussian distribution. Set the truncation radius to $\Lambda:=\frac{8Rm}{\eps}$ and define the corresponding truncation indicator by $
G_\Lambda(\mathbf{x}):=\Ind {\{\|\Phi(\mathbf{x})\|_2\le \Lambda\}}$.
For $\mathbf{c}\in\mathcal B_R$, define
$$
\ell_{\Lambda,\mathbf{c}}(\mathbf{x},y)
:=
G_\Lambda(\mathbf{x})\ell(y\langle \mathbf{c},\Phi(\mathbf{x})\rangle)
+
(1-G_\Lambda(\mathbf{x}))\ell_0.
$$
Let the population truncated objective be
$$
\cL^\Lambda_{\mu,\nu}(\mathbf{c})
:=
\E[\ell_{\Lambda,\mathbf{c}}(\mathbf{x},y)]
+
\mu\|\mathbf{c}\|_2^2
+
\nu\|\mathbf{A}(\mathbf{c})\|_*,
$$
and define the empirical truncated objective
$$
\widehat{\cL}^\Lambda_{\mu,\nu}(\mathbf{c})
:=
\frac1N\sum_{i=1}^N\ell_{\Lambda,\mathbf{c}}(\mathbf{x}_i,y_i)
+
\mu\|\mathbf{c}\|_2^2
+
\nu\|\mathbf{A}(\mathbf{c})\|_*.
$$
The algorithm computes $\widehat{\mathbf{c}}\in\mathcal B_R$ and outputs
$\widehat P=P_{\widehat{\mathbf{c}}}$, where $\widehat{\mathbf{c}}$ satisfies
\begin{equation}
\label{eq:approx-erm}
\widehat{\cL}^\Lambda_{\mu,\nu}(\widehat{\mathbf{c}})
\le
\inf_{\mathbf{c}\in\mathcal B_R}\widehat{\cL}^\Lambda_{\mu,\nu}(\mathbf{c})+\eps/4,
\end{equation}

We first compare the truncated and untruncated population objectives. For every $\mathbf{c}\in\mathcal B_R$,
\begin{align*}
\left|
\E\ell(y\langle \mathbf{c},\Phi(\mathbf{x})\rangle)
-
\E\ell_{\Lambda,\mathbf{c}}(\mathbf{x},y)
\right|
&\le
\E\big[
|\langle \mathbf{c},\Phi(\mathbf{x})\rangle|
\Ind {\{\|\Phi(\mathbf{x})\|_2>\Lambda\}}
\big] \\
&\le
R\,\E\big[
\|\Phi(\mathbf{x})\|_2
\Ind {\{\|\Phi(\mathbf{x})\|_2>\Lambda\}}
\big] \\
&\le
R\frac{\E\|\Phi(\mathbf{x})\|_2^2}{\Lambda}
=
\frac{Rm}{\Lambda}
=
\frac{\eps}{8}.
\end{align*}
The first inequality follows from the $1$-Lipschitzness of the logistic loss and the definition of the truncated loss; the second follows from the Cauchy--Schwarz inequality and $\|\mathbf{c}\|_2\le R$; the third uses Markov-type tail bound.

The regularizers are the same in $\cL_{\mu,\nu}$ and $\cL^\Lambda_{\mu,\nu}$, so
\begin{equation}
\label{eq:truncation-error}
\sup_{\mathbf{c}\in\mathcal{B}_R}
\left|\cL_{\mu,\nu}(\mathbf{c})-\cL^\Lambda_{\mu,\nu}(\mathbf{c})\right|
\le \frac{\eps}{8}.
\end{equation}
We next control the empirical error for the truncated loss class. Define $ h_{\mathbf{c}}(\mathbf{x},y):=\ell_{\Lambda,\mathbf{c}}(\mathbf{x},y)-\ell_0$. For $\mathbf{c}\in\mathcal B_R$, we have $|h_{\mathbf{c}}(\mathbf{x},y)|\le R\Lambda$,
$$
\E[h_{\mathbf{c}}(\mathbf{x},y)^2]\le R^2.
$$
Moreover, by symmetrization (\cite[Lemma 11.4]{BLM13}) and the contraction principle (\cite[Theorem 11.6]{BLM13}),
$$
\E\sup_{\mathbf{c}\in\mathcal B_R}
\left|
\frac1N\sum_{i=1}^Nh_{\mathbf{c}}(\mathbf{x}_i,y_i)-\E h_{\mathbf{c}}(\mathbf{x},y)
\right|
\le
C R\sqrt{\frac{m}{N}}.
$$
Applying Bousquet's version of Talagrand's inequality for bounded empirical processes, see \cite[Theorem~12.5]{BLM13}, gives that with probability at least $1-\delta$,
$$
\sup_{\mathbf{c}\in\mathcal B_R}
\left|
\frac1N\sum_{i=1}^N\ell_{\Lambda,\mathbf{c}}(\mathbf{x}_i,y_i)
-
\E\ell_{\Lambda,\mathbf{c}}(\mathbf{x},y)
\right|
\le
C\left(
R\sqrt{\frac{m}{N}}
+
R\sqrt{\frac{\log(e/\delta)}{N}}
+
R\Lambda\frac{\log(e/\delta)}{N}
\right).
$$
Since $\Lambda=8Rm/\eps$ and $R=O(1)$, the right-hand side is at most $\eps/8$ provided $N\ge C_\mu m\frac{\log(e/\delta)}{\eps^2}$. As $m=\binom{d+k}{k}\le (dk)^{O(k)}$, this is implied by
$$
N=(dk)^{O(k)}\frac{\log(e/\delta)}{\eps^2}.
$$
Consequently, with probability at least $1-\delta$,
\begin{equation}
\label{eq:truncated-uniform-convergence}
\sup_{\mathbf{c}\in \mathcal{B}_R}
\left|
\widehat{\cL}^\Lambda_{\mu,\nu}(\mathbf{c})-\cL^\Lambda_{\mu,\nu}(\mathbf{c})
\right|
\le \frac{\eps}{8}.
\end{equation}
Let $\mathbf{c}^\star\in\mathcal B_R$ be a minimizer of $\cL_{\mu,\nu}$ over $\R^m$. On the event above,
$$
\cL_{\mu,\nu}(\widehat{\mathbf{c}})
\overset{\eqref{eq:truncation-error}}{\le}
\cL^\Lambda_{\mu,\nu}(\widehat{\mathbf{c}})+\frac{\eps}{8} 
\overset{\eqref{eq:truncated-uniform-convergence}}{\le}
\widehat{\cL}^\Lambda_{\mu,\nu}(\widehat{\mathbf{c}})+\frac{\eps}{4} 
\overset{\eqref{eq:approx-erm}}{\le}
\widehat{\cL}^\Lambda_{\mu,\nu}(\mathbf{c}^\star)+\frac{\eps}{2} 
\overset{\eqref{eq:truncated-uniform-convergence}}{\le}
\cL^\Lambda_{\mu,\nu}(\mathbf{c}^\star)+\frac{3\eps}{4} 
\overset{\eqref{eq:truncation-error}}{\le}
\cL_{\mu,\nu}(\mathbf{c}^\star)+\eps.
$$
Therefore
$$
\cL_{\mu,\nu}(\widehat P)
\le
\inf_{P'\in\cP_k}\cL_{\mu,\nu}(P')+\eps.
$$

Finally, $\widehat{\cL}^\Lambda_{\mu,\nu}$ is a finite-dimensional convex objective over $\mathcal B_R$: the loss term is convex, $\mu\|\mathbf{c}\|_2^2$ is convex, and $\|\mathbf{A}(\mathbf{c})\|_*$ is convex because $\mathbf{A}$ is linear and the nuclear norm is convex (\Cref{app:lem:nuclear-hermite}). The dimension is $m\le (dk)^{O(k)}$, so an $\eps/4$-approximate minimizer can be found in time polynomial in the number of samples and the feature dimension, which is written as $\poly(N,d)$ with the dependence on $k$ absorbed into $(dk)^{O(k)}$.
\end{proof}

\section{FT-Mollification}\label{app:approx}

The univariate Gaussian $L_1$ approximation of a threshold by a degree-$O(\tau^{-2})$ polynomial is a standard consequence of the FT-mollification argument of \cite{DKN10}. 

\begin{theorem}[\cite{DKN10}]\label{prop:threshold-L1-standard}
There exists an absolute constant $C_{\mathrm{HS}}>0$ such that for every $\tau\in(0,1/10)$ and every $b\in\R$, the FT-mollification construction of \cite{DKN10} yields a polynomial $p:\R\to\R$ of degree at most $C_{\mathrm{HS}}\tau^{-2}$ satisfying
\begin{equation}
\label{q:taylor-good-approx}
\E_{x\sim\cN(0,1)}\bigl[\,|p(x)-\sign(x+b)|\,\bigr]\le \tau.
\end{equation}
\end{theorem}

In this section, we observe that the same construction also gives the $L_2$ and gradient bounds (\Cref{lem:threshold-L2-gradient}) needed in the main text.

The construction of \cite{DKN10}  starts with a mollified sign $g$, obtained by applying the
FT-mollifier at scale $\Theta(\tau)$. This mollified function $g$ is monotone non-decreasing bounded in $[-1,1]$ and satisfies\footnote{The first and last inequalities below are proven in \cite{DKN10}. The second inequality follows by bounding $\E_{x\sim\cN(0,1)}[(g'(x+b))^2] \leq C/\tau \cdot \E_{x\sim\cN(0,1)}[g'(x+b)]\leq C/\tau \int_{-\infty}^{+\infty} \frac{g'(x+b)}{\sqrt{2\pi}} dx \leq C/\tau $. Here, we used the fact that $g'$ is non-negative since $g$ is monotone non-decreasing, the penultimate inequality follows from  Gaussian anti-concentration, and the last inequality follows since $g$ is bounded in $[-1,1]$.}, for all $j \ge 0$,
$$
\sup_{x\in\R}|g^{(j)}(x+b)|\le (C/\tau)^j, \qquad 
\E_{x\sim\cN(0,1)}[(g'(x+b))^2]\le C\tau^{-1}.
$$
$$
\E_{x\sim\cN(0,1)}\bigl[\,|g(x+b)-\sign(x+b)|\,\bigr]\le \tau/2.
$$
Choosing $k=\Theta(\tau^{-2})$, \cite{DKN10} take $p$ to be the degree-$k$
Taylor polynomial at $x=0$ of the function $x\mapsto g(x+b)$:
$$
p(x):=
\sum_{j=0}^k \frac{g^{(j)}(b)}{j!}x^j.
$$
The work of \cite{DKN10} shows that
this polynomial satisfies Eq.\;\eqref{q:taylor-good-approx}:
To control the behavior of this polynomial outside the interval where the
Taylor approximation is accurate, we use the following standard Chebyshev
growth bound.

\begin{fact}[Chebyshev growth bound {\normalfont\cite[Theorem~2.20]{Riv74}}]\label{fact:cheb-growth}
Let $R,M>0$, and let $q$ be a polynomial of degree at most $k$. If $|q(x)|\le M$ for $|x|\le R,$ then for every $|x|\ge R$,
$$
|q(x)|\le M\left(\frac{2|x|}{R}\right)^k.
$$
\end{fact}

We now combine Taylor's theorem on a central interval with \Cref{fact:cheb-growth} on the tails to obtain the required $L_2$ and gradient bounds for the polynomial $p$.

\begin{lemma}\label{lem:threshold-L2-gradient}
Let $p$ be the degree-$k$ polynomial defined above, where $k=\Theta(\tau^{-2})$ is chosen with a sufficiently large implicit constant. Then
$$
\E_{x\sim\cN(0,1)}[p(x)^2]\le 5,
\qquad
\E_{x\sim\cN(0,1)}[(p'(x))^2]\le C\tau^{-1}.
$$
\end{lemma}

\begin{proof}
Choose a universal constant $\alpha\ge 4$ and set $R:=\alpha\sqrt{k}$. By Taylor's theorem and the derivative bounds for $g$,
$$
|g(x+b)-p(x)|
\le
\left(\frac{C|x|}{\tau k}\right)^{k+1},
$$
$$
|g'(x+b)-p'(x)|
\le
C\tau^{-1}
\left(\frac{C|x|}{\tau k}\right)^k.
$$
If $k\ge C\alpha^2\tau^{-2}$, then for all $|x|\le R$,
$$
|p(x)|\le 2,
\qquad
|g'(x+b)-p'(x)|\le C\tau^{-1}4^{-k}
\leq O(1).
$$
Hence
$$
\E_{x\sim\cN(0,1)}[p(x)^2\Ind{\{|x|\le R\}}]\le 4,
\qquad
\E_{x\sim\cN(0,1)}[(p'(x))^2\Ind{\{|x|\le R\}}]\le C\tau^{-1},
$$
where the second inequality uses $\E_{x\sim\cN(0,1)}[(g'(x+b))^2]\le C\tau^{-1}$ and $(a+b)^2\le 2a^2+2b^2$.

By \Cref{fact:cheb-growth}, for $|x|\ge R$,
$$
|p(x)|\le 2\left(\frac{2|x|}{R}\right)^k.
$$
Therefore, using $\E_{x\sim\cN(0,1)}[|x|^{2k}]=(2k-1)!!\le (2k)^k$ and taking $\alpha$ sufficiently large,
$$
\E_{x\sim\cN(0,1)}[p(x)^2]
\le
4+4\left(\frac{2}{R}\right)^{2k}\E_{x\sim\cN(0,1)}[|x|^{2k}]
\le
4+4\left(\frac{8}{\alpha^2}\right)^k
\le 5.
$$
Similarly, since $|p'(x)|\le C\tau^{-1}$ on $[-R,R]$, applying
\Cref{fact:cheb-growth} to $p'$ gives, for all $|x|\ge R$,
$$
|p'(x)|
\le
C\tau^{-1}\left(\frac{2|x|}{R}\right)^{k-1}.
$$
Using $\E_{x\sim\cN(0,1)}[|x|^{2m}]\le (2m)^m$ with $m=k-1$, we obtain
$$
\E_{x\sim\cN(0,1)}[(p'(x))^2\Ind{\{|x|>R\}}]
\le
C\tau^{-2}\left(\frac{8}{\alpha^2}\right)^{k-1}
\le
C\tau^{-1},
$$
for sufficiently large absolute constants in the choices of $\alpha$ and
$k=\Theta(\tau^{-2})$. Combining the estimates on $|x|\le R$ and $|x|>R$
gives
$$
\E_{x\sim\cN(0,1)}[(p'(x))^2]\le C\tau^{-1}.
$$
\end{proof}

\section{Multidimensional OU}\label{app:multi-ou}

We first prove a general approximation theorem in terms of Gaussian surface area, and then specialize it to intersections of halfspaces and Boolean functions of halfspaces.

\begin{theorem}[Gaussian $L_1$ approximation from Gaussian surface area]\label{thm:gsa-ou-approx}
There exists an absolute constant $C>0$ such that the following holds. Let $f:\R^d\to\{\pm1\}$ satisfy $\GSA(f)\ge 1$. For every $\tau\in(0,1/10)$, there exists a polynomial $S:\R^d\to\R$ of degree at most
$
C\frac{\GSA(f)^2\log(1/\tau)}{\tau^2}
$
such that
$$
\E_{\mathbf{x}\sim\cN(0,\mathbf{I}_d)}[|f(\mathbf{x})-S(\mathbf{x})|]\le \tau,
\qquad
\E_{\mathbf{x}\sim\cN(0,\mathbf{I}_d)}[S(\mathbf{x})^2]\le 1,
\qquad
\E_{\mathbf{x}\sim\cN(0,\mathbf{I}_d)}\|\nabla S(\mathbf{x})\|^2\le C\frac{\GSA(f)^2}{\tau}.
$$
\end{theorem}

\begin{proof}

Write the Hermite expansion
$
f=\sum_{\alpha}c_\alpha H_\alpha.
$
For $\rho\in(0,1)$, let $T_\rho$ denote the Ornstein--Uhlenbeck operator, so that for $\rho$-correlated standard Gaussians $\mathbf{x},\mathbf{y}$,
$$
T_\rho f(\mathbf{x})
:=
\sum_{\alpha}\rho^{|\alpha|}c_\alpha H_\alpha(\mathbf{x})
=
\E[f(\mathbf{y})\mid \mathbf{x}].
$$
Throughout the proof, expectations and probabilities involving both
$\mathbf{x}$ and $\mathbf{y}$ are taken over this joint $\rho$-correlated
Gaussian distribution.

By Jensen's inequality and the fact that $f$ is $\{\pm1\}$-valued,
$$
\E_{\mathbf{x}\sim\cN(0,\mathbf{I}_d)}[|f(\mathbf{x})-T_\rho f(\mathbf{x})|]
\le
\E[|f(\mathbf{x})-f(\mathbf{y})|]
=
2\Pr[f(\mathbf{x})\neq f(\mathbf{y})].
$$
By the Gaussian surface area noise-sensitivity bound \cite[Corollary~14]{KOS08},
$$
\Pr[f(\mathbf{x})\neq f(\mathbf{y})]
\le
C_0\,\GSA(f)\sqrt{1-\rho}.
$$
Choose
$
1-\rho=c\tau^2/\GSA(f)^2
$
for a sufficiently small absolute constant $c>0$. Then
$$
\E_{\mathbf{x}\sim\cN(0,\mathbf{I}_d)}[|f-T_\rho f|]\le \frac{\tau}{2}.
$$

Now truncate the Hermite expansion of $T_\rho f$ at degree $m$:
$$
S(\mathbf{x}):=\sum_{|\alpha|\le m}\rho^{|\alpha|}c_\alpha H_\alpha(\mathbf{x}).
$$
Since $f$ is $\{\pm1\}$-valued, Parseval's identity \Cref{fact:hermite-parseval} gives
$$
\E_{\mathbf{x}\sim\cN(0,\mathbf{I}_d)}[S(\mathbf{x})^2]
=
\sum_{|\alpha|\le m}\rho^{2|\alpha|}c_\alpha^2
\le
\sum_\alpha c_\alpha^2
=
\E[f(\mathbf{x})^2]
=
1.
$$
Also,
$$
\|T_\rho f-S\|_2^2
=
\sum_{|\alpha|>m}\rho^{2|\alpha|}c_\alpha^2
\le
\rho^{2m}\sum_\alpha c_\alpha^2
\le
\rho^{2m}.
$$
Hence
$$
\|T_\rho f-S\|_1
\le
\|T_\rho f-S\|_2
\le
\rho^m.
$$
Choose
$
m\ge \frac{\log(2/\tau)}{\log(1/\rho)}.
$
Since $\log(1/\rho)\ge 1-\rho$ for $\rho\in(0,1)$, it is enough to take
$$
m
=
O\!\left(\frac{\log(1/\tau)}{1-\rho}\right)
=
O\!\left(\frac{\GSA(f)^2\log(1/\tau)}{\tau^2}\right).
$$
Then $\|T_\rho f-S\|_1\le \tau/2$, and therefore
$$
\E_{\mathbf{x}\sim\cN(0,\mathbf{I}_d)}[|f(\mathbf{x})-S(\mathbf{x})|]
\le
\E_{\mathbf{x}\sim\cN(0,\mathbf{I}_d)}[|f(\mathbf{x})-T_\rho f(\mathbf{x})|]+\E_{\mathbf{x}\sim\cN(0,\mathbf{I}_d)}[|T_\rho f(\mathbf{x})-S(\mathbf{x})|]
\le
\tau.
$$

It remains to prove the gradient bound. Let
$
a_n:=\sum_{|\alpha|=n}c_\alpha^2
$
denote the Hermite weight of $f$ on level $n$. By \Cref{fact:hermite-derivative},
$$
\E_{\mathbf{x}\sim\cN(0,\mathbf{I}_d)}\|\nabla S(\mathbf{x})\|^2
=
\sum_{n=1}^m n\rho^{2n}a_n
\le
\sum_{n\ge 1}n\rho^{2n}a_n.
$$
For every $n\ge 1$,
$$
n\rho^{2n}
\le
\frac{1-\rho^n}{1-\rho}.
$$
Indeed,
$
(1-\rho^n)/(1-\rho)=1+\rho+\cdots+\rho^{n-1}\ge n\rho^{n-1}\ge n\rho^{2n}.
$
Therefore
$$
\E_{\mathbf{x}\sim\cN(0,\mathbf{I}_d)}\|\nabla S(\mathbf{x})\|^2
\le
\frac{1}{1-\rho}\sum_{n\ge1}(1-\rho^n)a_n.
$$
Since
$
\E[f(\mathbf{x})f(\mathbf{y})]=\sum_{n\ge0}\rho^n a_n
$
and
$
\E[f(\mathbf{x})^2]=1,
$
we have
$$
\sum_{n\ge1}(1-\rho^n)a_n
=
1-\E[f(\mathbf{x})f(\mathbf{y})]
=
2\Pr[f(\mathbf{x})\neq f(\mathbf{y})].
$$
Using the same Gaussian surface area noise-sensitivity bound,
$$
\sum_{n\ge1}(1-\rho^n)a_n
\le
C\,\GSA(f)\sqrt{1-\rho}.
$$
Thus
$$
\E_{\mathbf{x}\sim\cN(0,\mathbf{I}_d)}\|\nabla S(\mathbf{x})\|^2
\le
C\frac{\GSA(f)}{\sqrt{1-\rho}}.
$$
Substituting $1-\rho=\Theta(\tau^2/\GSA(f)^2)$ gives
$$
\E_{\mathbf{x}\sim\cN(0,\mathbf{I}_d)}\|\nabla S(\mathbf{x})\|^2
\le
C\frac{\GSA(f)^2}{\tau}.
$$
This completes the proof.
\end{proof}

We use the Gaussian surface area bounds recalled in
\Cref{fact:gsa-intersections} and \Cref{fact:gsa-boolean}: for intersections of $K$ halfspaces, $f\in\mathcal C_K^\cap$, one has $\GSA(f)\le O(\sqrt{\log K})$, while for arbitrary Boolean functions of $K$ halfspaces, $f\in\mathcal C_K$, one has $\GSA(f)\le O(K)$.

The following two corollaries are the approximation inputs used for intersections and Boolean functions of halfspaces.

\begin{corollary}[Gaussian $L_1$ approximation for intersections of halfspaces]\label{app:thm:intersection-approx}
There exist absolute constants $C_{\mathrm{IHS}},C'_{\mathrm{IHS}}>0$ with the following property. For every integer $K\ge 2$, every $\tau\in(0,1/10)$, and every $f\in \mathcal{C}_K^{\cap}$, there exists a polynomial $S:\R^d\to\R$ of degree at most
$
C_{\mathrm{IHS}}\frac{\log K\,\log(1/\tau)}{\tau^2}
$
such that, for $\mathbf{x}\sim\cN(0,\mathbf{I}_d)$,
$$
\E[|f(\mathbf{x})-S(\mathbf{x})|]\le \tau,
\qquad
\E[S(\mathbf{x})^2]\le 1,
\qquad
\E\|\nabla S(\mathbf{x})\|^2\le C'_{\mathrm{IHS}}\frac{\log K}{\tau}.
$$
\end{corollary}

\begin{proof}
By \Cref{fact:gsa-intersections}, $\GSA(f)\le O(\sqrt{\log K})$. Applying \Cref{thm:gsa-ou-approx} gives degree
$
O\!\left(\frac{\log K\,\log(1/\tau)}{\tau^2}\right)
$
and gradient bound
$$
\E\|\nabla S(\mathbf{x})\|^2\le O\!\left(\frac{\log K}{\tau}\right).
$$
The $L_1$ and $L_2$ bounds are those of \Cref{thm:gsa-ou-approx}.
\end{proof}

\begin{corollary}[Gaussian $L_1$ approximation for Boolean functions of halfspaces]\label{app:thm:boolean-approx}
There exist absolute constants $C_{\mathrm{BH}},C'_{\mathrm{BH}}>0$ with the following property. For every integer $K\ge 1$, every $\tau\in(0,1/10)$, and every $f\in\mathcal{C}_K$, there exists a polynomial $S:\R^d\to\R$ of degree at most
$
C_{\mathrm{BH}}\frac{K^2\log(1/\tau)}{\tau^2}
$
such that, for $\mathbf{x}\sim\cN(0,\mathbf{I}_d)$,
$$
\E[|f(\mathbf{x})-S(\mathbf{x})|]\le \tau,
\qquad
\E[S(\mathbf{x})^2]\le 1,
\qquad
\E\|\nabla S(\mathbf{x})\|^2\le C'_{\mathrm{BH}}\frac{K^2}{\tau}.
$$
Moreover, if
$
f=B(f_1,\dots,f_K)
$
where $f_j(\mathbf{x})=\sign(\langle \mathbf{w}_j,\mathbf{x}\rangle+t_j)$, then $S$ depends only on the span of the normal vectors $\mathbf{w}_1,\dots,\mathbf{w}_K$.
\end{corollary}

\begin{proof}
By \Cref{fact:gsa-boolean}, $\GSA(f)\le O(K)$. Applying \Cref{thm:gsa-ou-approx} gives degree
$
O\!\left(\frac{K^2\log(1/\tau)}{\tau^2}\right)
$
and gradient bound
$$
\E\|\nabla S(\mathbf{x})\|^2\le O\!\left(\frac{K^2}{\tau}\right).
$$
The $L_1$ and $L_2$ bounds are those of \Cref{thm:gsa-ou-approx}. If $f=B(f_1,\dots,f_K)$, then $f$ depends only on the span of the normal vectors $\mathbf{w}_1,\dots,\mathbf{w}_K$; the Ornstein--Uhlenbeck operator and Hermite truncation preserve this dependence.
\end{proof}

\section{Gaussian Poincaré inequality}\label{app:gp}

We restate and prove the following lemmas.

\begin{lemma}[Conditional Gaussian Poincar\'e]\label{app:lem:cond-poincare}
Let $\boldsymbol{\xi}\in\R^d$ be a unit vector and write $\mathbf{x}=\mathbf{u}+z\boldsymbol{\xi}$, where $\mathbf{u}\in\boldsymbol{\xi}^\perp$ and $z\sim\cN(0,1)$ are independent. For every weakly differentiable $G:\R^d\to\R$,
$$
\E\!\left[\Big(G(\mathbf{x})-\E[G(\mathbf{x})\mid \mathbf{u}]\Big)^2\right]
\le \E\big[(\partial_{\boldsymbol{\xi}} G(\mathbf{x}))^2\big].
$$
\end{lemma}

\begin{proof}
Condition on $\mathbf{u}\in\boldsymbol{\xi}^\perp$ and set $g_{\mathbf{u}}(z):=G(\mathbf{u}+z\boldsymbol{\xi})$. Then $\Var(G(\mathbf{x})\mid \mathbf{u})=\Var(g_{\mathbf{u}}(z))$. By the Gaussian Poincar\'e inequality~\cite[Theorem~3.20]{BLM13},
$$
\Var(g_{\mathbf{u}}(z))\le \E[(g_{\mathbf{u}}'(z))^2\mid \mathbf{u}].
$$
Since $g_{\mathbf{u}}'(z)=\partial_{\boldsymbol{\xi}}G(\mathbf{u}+z\boldsymbol{\xi})$, we get
$$
\Var(G(\mathbf{x})\mid \mathbf{u})\le \E[(\partial_{\boldsymbol{\xi}}G(\mathbf{x}))^2\mid \mathbf{u}].
$$
Taking expectation over $\mathbf{u}$ proves the claim.
\end{proof}

\begin{lemma}[Subspace Gaussian Poincar\'e]\label{app:subspace_gp}
Let $W\subseteq \R^d$ be a subspace of dimension $r$, and let $\{\boldsymbol{\xi}_1,\dots,\boldsymbol{\xi}_r\}$ be an orthonormal basis of $W$. Write $\mathbf{x}=\mathbf{u}+\mathbf{z}$, where $\mathbf{u}\in W^\perp$, $\mathbf{z}=\sum_{i=1}^r z_i\boldsymbol{\xi}_i$, and $z_1,\dots,z_r\stackrel{\mathrm{i.i.d.}}{\sim}\cN(0,1)$ are independent of $\mathbf{u}$. Then for every weakly differentiable $G:\R^d\to\R$,
$$
\E\!\left[\Big(G(\mathbf{x})-\E[G(\mathbf{x})\mid \mathbf{u}]\Big)^2\right]
\le \sum_{i=1}^r \E\big[(\partial_{\boldsymbol{\xi}_i}G(\mathbf{x}))^2\big].
$$
\end{lemma}

\begin{proof}
Condition on $\mathbf{u}\in W^\perp$ and set
$
g_{\mathbf{u}}(z_1,\dots,z_r):=G(\mathbf{u}+\sum_{i=1}^r z_i\boldsymbol{\xi}_i).
$
Then $\Var(G(\mathbf{x})\mid \mathbf{u})=\Var(g_{\mathbf{u}}(z_1,\dots,z_r))$. By the Gaussian Poincar\'e inequality~\cite[Theorem~3.20]{BLM13},
$$
\Var(g_{\mathbf{u}})\le \sum_{i=1}^r \E\!\left[\left(\frac{\partial g_{\mathbf{u}}}{\partial z_i}\right)^2 \middle| \mathbf{u}\right].
$$
Since $\frac{\partial g_{\mathbf{u}}}{\partial z_i}=\partial_{\boldsymbol{\xi}_i}G(\mathbf{x})$, we get
$$
\Var(G(\mathbf{x})\mid \mathbf{u})\le \sum_{i=1}^r \E[(\partial_{\boldsymbol{\xi}_i}G(\mathbf{x}))^2\mid \mathbf{u}].
$$
Taking expectation over $\mathbf{u}$ proves the claim.
\end{proof}

\end{document}